\documentclass[12pt,journal,draftclsnofoot,onecolumn]{IEEEtran}
\usepackage{amsfonts}
\usepackage{amssymb}
\usepackage{cite}
\usepackage[cmex10]{amsmath}
\usepackage{float}
\usepackage{color}
\usepackage{stfloats,fancyhdr}
\usepackage{amsmath}
\usepackage{algorithm}
\usepackage{algorithmic}
\usepackage{multirow}
\usepackage{changepage}
\usepackage[normalem]{ulem}
\usepackage{amsthm}

\newtheorem{theorem}{Theorem}
\newtheorem{lemma}{Lemma}

\newtheorem{definition}{Definition}

\newtheorem{remark}{Remark}

\IEEEoverridecommandlockouts

\ifCLASSINFOpdf
  \usepackage[pdftex]{graphicx}
  \DeclareGraphicsExtensions{.pdf,.jpeg,.png}
\else
  \usepackage[dvips]{graphicx}
  \DeclareGraphicsExtensions{.pdf}
\fi

\usepackage{subfigure}

\usepackage{fancybox,dashbox}
\begin{document}
\title{Optimizing Information Freshness in Two-Hop Status Update Systems under a Resource Constraint
%\thanks{This work was supported by the Fundamental Research Funds of Shandong University (2017TB0011), and the open research fund of National Mobile Communications Research Laboratory, Southeast University (2019D09).}
\thanks{The work of H. Chen was supported by the CUHK direct grant for research under the project code 4055126. The work of Yonghui Li was supported by ARC under Grant DP190101988 and DP210103410. The work of B. Vucetic was supported in part by the Australian Research Council Laureate Fellowship grant number FL160100032.}
\thanks{Y. Gu, Y. Li and B. Vucetic are with School of Electrical and Information Engineering, The University of Sydney, Sydney, NSW 2006, Australia (email: \{yifan.gu, yonghui.li, branka.vucetic\}@sydney.edu.au).}
\thanks{Q. Wang is with School of Electrical and Information Engineering, The University of Sydney, Sydney, NSW 2006, Australia, and Department of Information Engineering, The Chinese University of Hong Kong, Hong Kong SAR, China. The work was done when she was a visiting student at CUHK (email:qian.wang2@sydney.edu.au).}
\thanks{H. Chen is with Department of Information Engineering, The Chinese University of Hong Kong (e-mail: he.chen@ie.cuhk.edu.hk).
}
}
\author{Yifan Gu, Qian Wang, He Chen, Yonghui Li, and Branka Vucetic
}

\maketitle

\begin{abstract}
In this paper, we investigate the age minimization problem for a two-hop relay system, under a resource constraint on the average number of forwarding operations at the relay. We first design an optimal policy by modelling the considered scheduling problem as a constrained Markov decision process (CMDP) problem. Based on the observed multi-threshold structure of the optimal policy, we then devise a low-complexity double threshold relaying (DTR) policy with only two thresholds, one for relay's AoI and the other one for the age gain between destination and relay. We derive approximate closed-form expressions of the average AoI at the destination, and the average number of forwarding operations at the relay for the DTR policy, by modelling the tangled evolution of age at relay and destination as a Markov chain (MC). Numerical results validate all the theoretical analysis, and show that the low-complexity DTR policy can achieve near optimal performance compared with the optimal CMDP-based policy. Moreover, the relay should always consider the threshold for its local age to maintain a low age at the destination. When the resource constraint is relatively tight, it further needs to consider the threshold on the age gain to ensure that only those packets that can decrease destination's age dramatically will be forwarded.
\end{abstract}

\begin{IEEEkeywords}
Information freshness, age of information, status update, constrained Markov decision process, Markov chain.
\end{IEEEkeywords}


\IEEEpeerreviewmaketitle

\section{Introduction}
The next generation of wireless communication networks will support the connectivity of Internet of Things (IoT) devices, where one critical step is to deliver timely status update information for the underlying physical processes monitored by the IoT devices \cite{IoTbackground}. In the applications such as environment monitoring, vehicle tracking, smart parking, etc., the status update information generated from the end devices needs to be kept as fresh as possible from the perspective of the receiver side. To quantize the freshness and timeliness of the information, a new performance metric, termed the Age of Information (AoI), has been proposed recently \cite{AoIbackground}. The AoI is defined as the time elapsed since the generation of the last successfully received status update. Different from the conventional performance metrics such as delay and throughput, the AoI captures both the latency and the generation time for each status update from the receiver's perspective. Consider a point-to-point status update system as an example, where the source samples a time-varying physical process and transmits the sampled status update to a destination. Due to the error-prone channel that introduces delay, it takes time for the status update to be successfully received by the destination. If the most recently received status update carries the information sampled at time $r\left(t\right)$, the status update age at time $t$ from the perspective of destination is defined as $t- r\left(t\right)$ \cite{AoIbackground}.
\subsection{Background}
The minimization of AoI has been studied extensively in the existing literature for various system setups and assumptions. The seminal works on AoI focused on the queueing theory-based studies for point-to-point systems \cite{AoIbackground, AoI_LCFS, AEphremides_management, AEphremides_MM12_deadline, Updateorwait}. For instance, the authors in \cite{AoIbackground} derived the average AoI of a First-Come-First-Serve (FCFS) queueing system and showed that there exists an optimal generation rate for the status updates to minimize the average AoI. The optimal generation rate is different from those maximize the throughput or minimize the delay. Reference \cite{AoI_LCFS} studied a Last-Come-First-Serve (LCFS) queueing system and concluded that compared with the FCFS, the LCFS queueing system can reduce the average AoI by always transmitting the latest generated status updates. The minimization of AoI for more complicated networks were studied for cognitive-radio (CR)-based networks \cite{AValehi_CR, SLeng_CR, YifanCR, QianCR}, and for networks with multiple sources or destinations \cite{Ryates_multiplesources,R.Talak-Centralized-UnknownCSI-MAC,I.Kadota-Centralized-MAC,Z.Jiang-Decentralized,S.Kaul-Distributed-Centralized-MAC,yifanmac}. The optimal sampling strategy at the source to minimize the average age at the destination was studied in \cite{Rnew}, where the optimization problem was modelled by a constrained Markov decision process. The age optimal policies for energy limited devices were studied in \cite{AoI_EH1,AoI_EH2,AoI_EH3,AoI_EH4,AoI_EH5,AoI_EH6} for energy harvesting nodes, and in \cite{YIFANAOI, AoI_HARQ, BzhouCMDP} where a constraint on the long-term average number of transmissions was considered. All aforementioned works considered on the single-hop networks.

There have been some efforts on studying the AoI minimization of multi-hop networks, where the status update delivery between source and destination is assisted by relay(s) \cite{multihop1,multihop2,multihop3,multihop4,multihop5,simplerelay,bobrelay,R1,R2,R3,R4}. The authors in \cite{multihop1} considered a general multi-hop networks and pointed out that the non-preemptive LCFS policy can minimize the age at all nodes among all non-preemptive and work-conserving policies. Reference \cite{multihop2} derived optimal scheduling policies for a multi-hop networks under a general interference constraint. Reference \cite{multihop3} developed an age-based scheduler that uses the information of both the instantaneous age and the interarrival times of packets in making decisions. Asymptotic analysis was also performed to derive age upper and lower bounds for the considered age-based scheduler. The authors in \cite{multihop4} characterized lower bound for the average AoI of the multi-hop networks. A near-optimal scheduling algorithm was also derived and it was shown that the developed algorithm can achieve the lower bound within an additive gap scaling linearly with the size of the network. References \cite{multihop5,simplerelay,bobrelay} focused on the performance analysis of the average AoI for multi-hop systems. Reference \cite{multihop5} studied the AoI performance of a random decision-based schedule in gossip networks. Three performance analysis methods, including explicit analysis, algorithm-based analysis and asymptotic approximations, were proposed to characterize the average AoI. The authors in \cite{simplerelay} visited a simple relay system where the relay can generate age-sensitive packets. Besides, the relay needs to forward packets generated by another stream with priorities lower or higher than its own packets. Under this scenario, the average AoI of the relay's status update packets was analyzed by resorting to queueing theory. \cite{bobrelay} designed a retransmission scheme for a three node relay system with random arrival of status updates at source. The average AoI of the proposed retransmission scheme was then analyzed. The authors in \cite{R1} studied the age distribution for a two-hop system and minimized the tail of the AoI distribution by tuning the generation frequency of the status updates. M. Moradian and A. Dalani investigated two relay networks, one with direct link and the other one without direct link. The average AoI of the considered systems were then analyzed \cite{R2}. The authors in \cite{R3} derived the average AoI for tandem queues where the first queue follows FCFS while the second queue has a capacity of 1. Very recently, the authors in \cite{R4} studied the distribution of the Peak Age of Information (PAoI) for a system with two tandem $M/M/1$ queues.
\subsection{Motivation and Contributions}
All the aforementioned works on the multi-hop systems \cite{multihop1,multihop2,multihop3,multihop4,multihop5,simplerelay,bobrelay,R1,R2,R3,R4} assumed that there is no resource constraint\footnote{Note that \cite{multihop2,multihop3,multihop4,multihop5} considered a resource constraint from an interference-limited perspective while the resource constraint considered in this paper is from an energy consumption perspective.} in terms of energy consumption on the intermediate nodes. In practical system implementation, it is highly undesirable that the nodes keep transmitting the status updates due to interference or energy constraints \cite{YIFANAOI, AoI_HARQ, BzhouCMDP}. Besides, policies with simple structures and the corresponding closed-form expression of the average AoI are also crucial in practical system design. On the one hand, theoretical insights can be gained by the closed-form analytical results. On the other hand, they can be used to quickly verify whether the required AoI performance can be guaranteed or not with a given set of system parameters. In this context, several natural and fundamental questions arise: \textit{What is the optimal policy for a multi-hop system with a resource constraint? Is the optimal policy having a simple structure, or can we find simple yet near-optimal policies? What is the average AoI performance of the optimal or near-optimal policies in a multi-hop system under a resource constraint?} As an initial effort to answer these important questions, in this paper, we investigate a two-hop status update system consisting of one source, one destination and one relay. We consider a slotted communication scenario, where the relay can choose either receiving operation or forwarding operation at the beginning of each time slot. The goal of this paper is to derive optimal scheduling policies to be implemented at the relay to minimize the average AoI at the destination under a resource constraint for the average number of forwarding operations at the relay. The main contributions of this paper are summarized as follows:
\begin{itemize}
  \item We first derive an optimal scheduling policy by modeling the considered two-hop status update system as a constrained Markov decision process (CMDP) problem. To solve the CMDP problem and obtain the optimal scheduling policy, we apply the Lagrangian relaxation and transform the CMDP problem into an unconstrained Markov decision process (MDP). Structural analysis is then performed for the optimal scheduling policy to unveil that it has a multiple threshold structure.
  \item Based on the structure of the CMDP-based policy, we then devise a low-complexity policy, namely the double threshold relaying (DTR). In DTR, only two thresholds are adopted, one for the relay's age and the other one for the age gain between destination and relay. Different from most of the works on the scheduling of multi-hop systems that only provided simulation results, in order to gain some theoretical insights, we aim to characterize the average AoI performance of the proposed DTR policy in a closed-form. Specifically, we first obtain the stationary distribution of the DTR policy in terms of the instantaneous age of the relay and destination by applying a Markov chain (MC)-based method. Note that the evolution of the states in the MC is not straightforward due to the adoption of two thresholds, making the analysis for the stationary distribution non-trivial.
  \item Based on the derived stationary distribution of the instantaneous age, we characterize closed-form expressions for the average AoI at the destination, and the average number of forwarding operations at the relay for the proposed DTR policy. The analytical expressions can be used to optimize the two thresholds in the DTR policy efficiently. Besides, they reveal some insights on practical system designs, and the achievable average AoI performance of the DTR policy under a given resource constraint. Simulation results verify all the theoretical analysis, and show that the proposed DTR policy can achieve near-optimal performance compared with the optimal CMDP-based policy. It also indicates that only one threshold for the relay's age is needed in the DTR policy when there is no resource constraint or the resource constraint is loose. This is because that the relay should always forward the status updates with relatively low age by considering the threshold on its local AoI. When the resource constraint becomes tight, the relay needs to further decrease its number of transmissions by the threshold of the age gain, such that it only forwards those low-age status updates that can also decrease destination's age dramatically. %is more preferable for a system with a lower generation rate at the PIoT, a lower peak transmit power constraint of the SIoT, and a higher interference between the PIoT and SIoT.
\end{itemize}
%The main differences between our work and references \cite{multihop1,multihop2,multihop3,multihop4} are The authors in \cite{multihop4} studied a multi-hop network and characterized lower bound for the average AoI of the considered network. A near optimal scheduling algorithm is also derived and it was shown that the developed algorithm can achieve the lower bound within an additive gap scaling linearly with the size of the network.
It is worth emphasizing that the authors in \cite{EHtwohop} also considered a similar two-hop status update system with a resource constraint where both the source and relay are relying on the harvested energy from the nature. Based on the different arrival patterns of the energy at source and relay, optimal online and offline scheduling policies were designed for the considered system. Our work is fundamentally different from \cite{EHtwohop} on two aspects. First of all, despite the optimal CMDP-based policy, we also propose a low-complexity DTR policy for the considered system which is easy to be implemented and has a structure with only two thresholds. Secondly, we characterize the average AoI of the proposed DTR policy in a closed-form, while no theoretical analysis was provided in \cite{EHtwohop}.

We further point out that the performance analysis of the proposed DTR policy in terms of the stationary distribution, average AoI, and average number of forwarding operations is challenging and different from the existing schemes in the aforementioned literature. Specifically, the evolution of the instantaneous age becomes rather complicated when having two thresholds. The instantaneous age at both relay and destination are tangled together and the conventional performance analysis methods become intractable in our considered case where new performance analysis method is required. To the best knowledge of the authors, this is the first paper to study the age distribution and average AoI performance of a policy with two age thresholds. The proposed performance analysis method and the derived results in this paper remain general and incorporate the existing single threshold policies as special cases. At last, our proposed analytical method is not limited to the considered system model and can potentially be applied in other system models where schemes with two thresholds are considered.
\subsection{Organization}
The rest of the paper is organized as follows. Section II introduces the system model. We summarize the main results, insights, and Theorems provided in this paper in Section III. Sections IV, V and VI provide the proofs for the Theorems presented in Section III. Numerical results are presented in Section VII to validate the theoretical analysis and compare the CMDP-based and DTR policies. Finally, conclusions are drawn in Section VIII.

\textbf{\emph{Notation}}: Throughout this paper, $\mathbb{E}\left[ A  \right]$ is the expectation of $A$ and $\Pr\left\{A \right\}$ is the probability of $A$. $\left\lfloor \cdot \right\rfloor$ is the floor function and $\bmod$ is the modulo operation.

\section{System Model and Problem Formulation}
%\begin{figure}
%\centering \scalebox{0.4}{\includegraphics{Fig1.eps}}
%\caption{The considered two-hop status update system. }\label{fig:systemmodel}
%\end{figure}
\subsection{System Description}
We consider a relay-assisted two-hop status update system consisting of one source node $S$, one decode-and-forward relay node $R$, and one destination node $D$. We assume that the direct link between $S$ and $D$ does not exist due to the long distance and high attenuation. For implementation simplicity and low cost, all nodes in the considered status update system are assumed to be single-antenna and half-duplex devices. We study a slotted communication scenario where time is divided into slots with equal durations, and each status update transmission for the $S-R$ and $R-D$ links occupy one time slot. Furthermore, the links $S-R$ and $R-D$ are considered to be error-prone, and the successful transmission probabilities of the $S-R$ and $R-D$ links are denoted by $p$ and $q$, respectively.

In this paper, we assume that $S$ adopts the generate-at-will model such that it can generate a fresh status update at the beginning of each time slot. Furthermore, we assume that $R$ can only store one decoded status update from $S$, and $R$ discards the stale status update when a new status update is decoded correctly from $S$. In order for $R$ to track the instantaneous AoI at $D$, it is assumed that $D$ can provide delay-free and error-free feedback to $R$ to indicate the transmission result for the transmissions over the $R-D$ link. According to the half-duplex model at the relay and with the instantaneous AoIs at $R$ and $D$, $R$ can choose one of the following two actions at the beginning of each time slot: receiving from $S$ or forwarding to $D$. In the receiving action, $R$ polls a transmission from $S$ and receives the status update transmitted by $S$. If the decoding is correct, $R$ stores the status update in its buffer. In the forwarding action, $R$ forwards the stored status update to $D$ by using one time slot.
\subsection{Age of Information}
We use $w \left(t\right) \in \mathcal{W}$ to represent the action taken by $R$ in the time slot $t$, where $\mathcal{W} \triangleq \left\{0,1\right\}$ is the action space. Specifically, $w \left(t\right)=0$ and $w \left(t\right)=1$ represent the receiving action and the forwarding action, respectively. We let $a_S \left(t\right)$, $a_R \left(t\right)$ and $a_D \left(t\right)$ denote the instantaneous AoI of $S$, $R$ and $D$, respectively, at the beginning of each time slot $t$. Furthermore, we use $g\left(t\right) = a_D \left(t\right) - a_R \left(t\right)$ to be the instantaneous age gain in the time slot $t$. Because $S$ adopts the generate-at-will model, we have $a_S\left(t\right) = 0, \forall t$. We then use $I_{SR} \left(t\right)$ and $I_{RD}\left(t\right)$ to be the indicator for the $S-R$ and $R-D$ links at time slot $t$, to indicate whether the transmission will be successful or not. We have $\Pr \left\{I_{SR} \left(t\right) = 1\right\} = p$, $\Pr \left\{I_{SR} \left(t\right) = 0\right\} = 1-p$, $\Pr \left\{I_{RD} \left(t\right) = 1\right\} = q$, and $\Pr \left\{I_{RD} \left(t\right) = 0\right\} = 1-q$, respectively. With the above definitions, the evolution of the instantaneous AoI at $R$, and the instantaneous age gain can be expressed as
\begin{equation}\label{AoIR}
a_R\left(t+1\right) =
\left\{ \begin{matrix}
\begin{split}
   &{a_S\left(t\right)+1, \quad w\left(t\right)=0 \cap I_{SR}\left(t\right) = 1},\\
   &{a_R\left(t\right)+1, \quad \text{Otherwise}},\\
\end{split}
\end{matrix}
\right.
\end{equation}
\begin{equation}\label{AoID}
\begin{split}
&g\left(t+1\right) =\left\{ \begin{matrix}
\begin{split}
   &{0, \quad w\left(t\right)=1 \cap I_{RD}\left(t\right) = 1},\\
   &{a_R\left(t\right) + g\left(t\right), \quad w\left(t\right)=0 \cap I_{SR}\left(t\right) = 1},\\
   &{g\left(t\right), \quad \text{Otherwise}}.
\end{split}
\end{matrix}
\right.
\end{split}
\end{equation}
\subsection{Optimization Problem Formulation}
We realize that successive transmission, i.e., keep choosing forwarding action at the relay\footnote{In this paper, we only consider the resource constraint at $R$ for a specific relay link and focus on the design and analysis of relaying policies. Note that the status updates generated from $S$ could target at multiple destinations and the resource constraint at $S$ is thus not considered for the one link model.}, is typically undesirable in practical systems due to energy or interference constraints. We thus follow \cite{YIFANAOI, AoI_HARQ, BzhouCMDP} and introduce a resource constraint on the average number of forwarding actions. We require that the long-term average number of forwarding actions cannot exceed $\eta_C$. In other words, the long-term probability of choosing forwarding actions at the relay cannot exceed $\eta_C$ and $0 < \eta_C \le 1$. We now give some important definitions before mathematically describe the optimization problem. We define $v\triangleq \left(k,d\right) \in \mathcal{V}$ as the system state of the considered status update system, where $k$ is the instantaneous AoI at $R$, and $d$ is the instantaneous age gain. The state space $\mathcal{V}$ is given in the following Lemma.
\begin{lemma}\label{lemmastatespace}
The state space $\mathcal{V}$ is given by
\begin{equation}\label{statespace}
\mathcal{V} = \left\{\left(k,d\right): \left(k \ge 2, d = 0\right)  \cup \left( k \ge 1, d \ge 2\right), k,d \in \mathbb{N} \right\}.
\end{equation}
\end{lemma}
\begin{proof}
{See Appendix A.}
\end{proof}
From Lemma \ref{lemmastatespace}, we can observe that the age gain of the considered system cannot be $1$ in all states. Besides, when the relay and destination are synchronized, i.e., the age gain is $0$, the instantaneous age at the relay is no less than $2$. Throughout the analysis provided in this paper, we only consider the states within the state space $\mathcal{V}$ because the states out of $\mathcal{V}$ are not reachable. With the defined system states and actions, we then give the definition of the considered stationary policies.
\begin{definition}(Stationary Policy)\label{def1}
A stationary policy $\theta$ is defined as a mapping from the system state $\left(k,d\right) \in \mathcal{V}$ to the action $w \in \mathcal{W}$, where $\theta\left(k,d\right) = w$. The action is only dependent on the current states, and independent to the past states and actions, and time invariant.
\end{definition}
For a given stationary policy $\theta$, from (\ref{AoIR}) and (\ref{AoID}), the average AoI of the considered two-hop status update system, and the average number of forwarding actions are given by
\begin{equation}\label{averageAoIsystem}
\overline \Delta \left({\theta}\right)   = \mathop {\lim \sup }\limits_{T \to \infty } {1 \over T}\mathbb{E}\left[ {\sum\limits_{t = 1}^T {\left[ {{a_R}\left( t \right) + g\left( t \right)} \right]} } \right],
\end{equation}
\begin{equation}\label{longtermenergy}
\eta\left(\theta\right) = \mathop {\lim \sup }\limits_{T \to \infty } {1 \over T}\mathbb{E}\left[ {\sum\limits_{t = 1}^T {w\left( t \right)} }\right],
\end{equation}
where the expectation operation is introduced by the dynamic of system states. The age minimization problem can now be described by
\begin{equation}\label{optimizationP}
	\begin{aligned}
	& \text{Minimize} \quad \overline \Delta \left({\theta}\right),\\
	& {\text{Subject to}}  \quad \eta\left(\theta\right) \le \eta_c.
	\end{aligned}
\end{equation}
%In the following sections, we analyze the average AoI of the proposed relaying protocols by evaluating the three expectation terms $\mathbb{E} \left[{S}\right]$, $\mathbb{E} \left[{Y}\right]$ and $\mathbb{E} \left[{Y^2}\right]$ in (\ref{AoIexpression}).

\section{Main Results}
This section presents the main results of this paper, including the design of two different policies for the considered status update system, denoted by $\theta_{\rm{CMDP}}$ and $\theta_{\rm{DTR}}$. Specifically, $\theta_{\rm{CMDP}}$ is the optimal stationary policy for the formulated problem (\ref{optimizationP}), and it is derived by adopting a CMDP method. We show that $\theta_{\rm{CMDP}}$ has a multiple threshold structure and it is intractable to characterize the average AoI performance in a closed-form. In order to reduce the implementation complexity and leverage the derived structural results of the CMDP-based policy, we then propose a low-complexity double threshold relaying (DTR) policy $\theta_{\rm{DTR}}$ based on the observed structure of $\theta_{\rm{CMDP}}$. The policy $\theta_{\rm{DTR}}$ can be implemented easily with two individual thresholds and has a near-optimal performance compared with $\theta_{\rm{CMDP}}$. Closed-form expressions for the age distribution, average AoI and average number of forwarding actions of $\theta_{\rm{DTR}}$ policy are also derived. The main results for $\theta_{\rm{CMDP}}$ (Theorems \ref{theoremCMDP1}-\ref{theoremCMDP2}) and $\theta_{\rm{DTR}}$ (Theorems \ref{TheoremDistribution}-\ref{TheoremTP}) are summarized in the following subsections. The proofs of Theorems \ref{theoremCMDP1}-\ref{theoremCMDP2} will be given in Section IV, the proof of Theorem \ref{TheoremDistribution} will be given in Section V, and the proofs of Theorems \ref{TheoremAoI}-\ref{TheoremTP} will be presented in Section VI. Insightful results from our analytical results are also discussed in this section.
\subsection{Optimal Stationary Policy $\theta_{\rm{CMDP}}$}
\begin{theorem}\label{theoremCMDP1}
The optimal stationary policy for problem (\ref{optimizationP}) $\theta_{\rm{CMDP}}$ randomizes in at most one state. Specifically, the optimal policy can be expressed as a randomized mixture of two stationary deterministic policies $\theta_{\lambda_1^*}$ and $\theta_{\lambda_2^*}$, mathematically,
\begin{equation}\label{theorem1eq}
\theta_{\rm{CMDP}} = \alpha \theta_{\lambda_1^*} + \left(1 - \alpha\right) \theta_{\lambda_2^*},
\end{equation}
where $\alpha \in \left[0,1\right]$ is the randomization parameter.
\end{theorem}
Theorem \ref{theoremCMDP1} shows that the optimal policy may not be deterministic, which may involve the mixture of at most two deterministic policies. Moreover, when $\alpha$ equals to $0$ or $1$, the optimal policy is deterministic. Note that the detailed distribution of the randomized mixture in terms of $\alpha$, and the solutions to $\theta_{\lambda_1^*}$ and $\theta_{\lambda_2^*}$ are discussed in the proof given in Section IV.
\begin{theorem}\label{theoremCMDP2}
The structure of each stationary deterministic policy $\theta_{\lambda^*}$, ${\lambda^*} \in \left\{\lambda_1^*, \lambda_2^*\right\}$ which forms $\theta_{\rm{CMDP}}$, has a switching-type given by
 	\begin{itemize}
 		\item If $\theta_{\lambda^*}\left(k,d\right)=1$, then $\theta_{\lambda^*}(k,d+z)=1$, $\forall z \in \mathbb{N}$,
 		\item If $\theta_{\lambda^*}(k,d)=0$, then $\theta_{\lambda^*}(k+z,d)=0$, $\forall z \in \mathbb{N} $.
 	\end{itemize}
\end{theorem}
We can observe from Theorem \ref{theoremCMDP2} that the policy $\theta_{\lambda^*}$ may have a multiple threshold structure. More importantly, we realize that $\theta_{\lambda^*}$ tends to select forwarding action under two conditions: (1) $R$'s instantaneous AoI, i.e., $k$, is relatively low; (2) The age gain, i.e., $d$, is high is relatively high. The first condition is understandable because $R$ should only forward fresh status updates to $D$ and should receive new status updates from $S$ when the stored status update becomes stale. The rationale behind the second condition is that $R$ should only forward those status updates that can decrease $D$'s instantaneous AoI dramatically. This becomes especially important when a resource constraint is considered at $R$.

Furthermore, it is worth pointing out that when the resource constraint $\eta_C$ is loose, the considered problem can be solved by MDP. Specifically, let $\theta_{MDP}$ denote the optimal MDP policy without considering the resource constraint and $\eta\left(\theta_{MDP}\right)$ be the corresponding average number of forwarding actions. The optimal policy for problem (\ref{optimizationP}) will be $\theta_{MDP}$ when the required resource constraint $\eta\left(\theta_{MDP}\right) \le \eta_C \le 1$. This is understandable because the optimal policy without considering the resource constraint has already satisfied $\eta_C$, and $\eta\left(\theta_{MDP}\right)$ can be evaluated numerically when deriving $\theta_{MDP}$. On the other hand, if the required resource constraint $0< \eta_C < \eta\left(\theta_{MDP}\right)$, the problem (\ref{optimizationP}) needs to be solved by the CMDP model. Because the analysis of the $\theta_{MDP}$ is very similar to the $\theta_{CMDP}$ and also follows the structural results given in Theorem \ref{theoremCMDP2}. For the purpose of brevity, we will not provide the detailed analysis for the $\theta_{MDP}$ policy.
\subsection{Low-Complexity Stationary Deterministic Policy $\theta_{\rm{DTR}}$}
Based on the observation from Theorem \ref{theoremCMDP2}, in order to simplify the policy structure, we are motivated to implement two thresholds, $\delta_1$ for $R$'s instantaneous AoI, and $\delta_2$ for the instantaneous age gain, to efficiently decide the operation mode for the status update system. In the proposed $\theta_{\rm{DTR}}$ policy, $R$ chooses the forwarding action only when its instantaneous AoI is no greater than $\delta_1$, and the age gain is no less than $\delta_2$. Otherwise, $R$ selects the receiving action. The DTR policy can be described as
\begin{equation}\label{DTR}
\begin{split}
&\theta_{\rm{DTR}}\left(k,d\right) =\left\{ \begin{matrix}
\begin{split}
   &{1, \quad k \le \delta_1 \cap d \ge \delta_2},\\
   &{0, \quad \text{Otherwise}}.
\end{split}
\end{matrix}
\right.
\end{split}
\end{equation}
Compared with the $\theta_{\rm{CMDP}}$ policy, the proposed $\theta_{\rm{DTR}}$ policy can be implemented easily by using two thresholds $\delta_1$ and $\delta_2$. Two natural question arises: \textit{What is the theoretical performance in terms of average AoI for the $\theta_{\rm{DTR}}$ policy? How can we tune $\delta_1$ and $\delta_2$ to minimize the average AoI at $D$ and satisfy the resource constraint $\eta_c$?} In order to answer these questions, we first present Theorem \ref{TheoremDistribution} for the joint distribution of $k$ and $d$ for the $\theta_{\rm{DTR}}$ policy.
\begin{theorem}\label{TheoremDistribution}
Let $\pi_{k,d}$ denote the stationary distribution of state $v\triangleq \left(k,d\right)$, $v \in \mathcal{V}$. The stationary distribution for the $\theta_{\rm{DTR}}$ policy can be summarized in (\ref{stationaryMC1}) for the case that $\delta_1 \ge \delta_2-1$, and in (\ref{stationaryMC2}) for the case that $\delta_1 \le \delta_2-1$.
\begin{subequations}\label{stationaryMC1}
\begin{align}
&\pi_{k,d} =
\left\{{
\begin{matrix}
\begin{split}
   &{{{{{\left( {1 - p} \right)}^{k - 1}} - {{\left( {1 - q} \right)}^{k - 1}}} \over {q - p}}x, 2 \le k\le\delta_1+1, d=0,} \\
   &{{{\left[ {{{\left( {1 - p} \right)}^{{\delta _1}}} - {{\left( {1 - q} \right)}^{{\delta _1}}}} \right]{{\left( {1 - p} \right)}^{k - {\delta _1} - 1}}} \over {q - p}}x, k\ge \delta_1+1,d=0,} \\
\end{split}
\end{matrix}
}\right. \\
&\pi_{k,d} =
   {{{p{{\left( {1 - p} \right)}^{k - 1}}\left[ {1 - {{\left( {1 - q} \right)}^{d - 1}}} \right]} \over q}x, \forall k, 2\le d \le \delta_2,}  \\
&\pi_{k,d} \approx
\left\{{
\begin{matrix}
\begin{split}
   &{\left( {1 - q} \right)^{k-1}{\pi_{1,\delta_2}}\sum\limits_{l = 0}^n {{{\left( {N+l-1} \right)!} \over {l!\left( N-1 \right)!}}} {\left[ {p{{\left( {1 - q} \right)}^{{\delta _1}}}} \right]^l}{\left( {1 - p} \right)^{N-1}},1 \le k \le \delta_1+1, d \ge \delta_2, }\\
   &{\left( {1 - q} \right)^{\delta_1}\left(1-p\right)^{k-\delta_1-1}{\pi_{1,\delta_2}}\sum\limits_{l = 0}^n {{{\left( {N+l-1} \right)!} \over {l!\left( N-1 \right)!}}} {\left[ {p{{\left( {1 - q} \right)}^{{\delta _1}}}} \right]^l}{\left( {1 - p} \right)^{N-1}}, k \ge \delta_1+1, d \ge \delta_2,}
\end{split}
\end{matrix}
}\right.
\end{align}
\end{subequations}
\begin{subequations}\label{stationaryMC2}
\begin{align}
&\pi_{k,d} =
\left\{{
\begin{matrix}
\begin{split}
   &{{{{{\left( {1 - p} \right)}^{k - 1}} - {{\left( {1 - q} \right)}^{k - 1}}} \over {q - p}}x, 2 \le k\le\delta_1+1, d=0,} \\
   &{{{\left[ {{{\left( {1 - p} \right)}^{{\delta _1}}} - {{\left( {1 - q} \right)}^{{\delta _1}}}} \right]{{\left( {1 - p} \right)}^{k - {\delta _1} - 1}}} \over {q - p}}x, k\ge \delta_1+1,d=0,} \\
\end{split}
\end{matrix}
}\right. \\
&\pi_{k,d} =
\left\{{
\begin{matrix}
\begin{split}
   &{{{p{{\left( {1 - p} \right)}^{k - 1}}\left[ {1 - {{\left( {1 - q} \right)}^{d - 1}}} \right]} \over q}x, \forall k,2 \le d \le \delta_1+1,}  \\
   &{{{p{{\left( {1 - p} \right)}^{k - 1}}\left[ {1 - {{\left( {1 - q} \right)}^{\delta_1}}} \right]} \over q}x, \forall k,\delta_1+1 \le d \le \delta_2,}\\
\end{split}
\end{matrix}
}\right.  \\
&\pi_{k,d} =
\left\{{
\begin{matrix}
\begin{split}
   &{\left( {1 - q} \right)^{k-1}{\pi_{1,\delta_2}}\sum\limits_{l = 0}^n {{{\left( {N+l-1} \right)!} \over {l!\left( N-1 \right)!}}} {\left[ {p{{\left( {1 - q} \right)}^{{\delta _1}}}} \right]^l}{\left( {1 - p} \right)^{N-1}},1 \le k \le \delta_1+1, d \ge \delta_2, }\\
   &{\left( {1 - q} \right)^{\delta_1}\left(1-p\right)^{k-\delta_1-1}{\pi_{1,\delta_2}}\sum\limits_{l = 0}^n {{{\left( {N+l-1} \right)!} \over {l!\left( N-1 \right)!}}} {\left[ {p{{\left( {1 - q} \right)}^{{\delta _1}}}} \right]^l}{\left( {1 - p} \right)^{N-1}}, k \ge \delta_1+1, d \ge \delta_2.}
\end{split}
\end{matrix}
}\right.
\end{align}
\end{subequations}
The term $x$ in (\ref{stationaryMC1}) and (\ref{stationaryMC2}) is different for the two cases $\delta_1 \ge \delta_2-1$ and $\delta_1 \le \delta_2-1$ given by
\begin{equation}\label{x}
x =\left\{{
\begin{matrix}
\begin{split}
   &{{p{q^2}} \over {q\left( {1 - p} \right) + pq{\delta _2} + p{{\left( {1 - q} \right)}^{{\delta _2} - 1}}\left[ {1 - {{\left( {1 - q} \right)}^{{\delta _1} - {\delta _2} + 1}}} \right]}}, \delta_1 \ge \delta_2-1, \\
   &{{pq} \over {\left( {1 - p} \right) + p{\delta _2} - p\left( {{\delta _2} - {\delta _1} - 1} \right){{\left( {1 - q} \right)}^{{\delta _1}}}}},\delta_1 \le \delta_2-1. \\
\end{split}
\end{matrix}
}\right.
\end{equation}
The term $\pi_{1,\delta_2}$ in (\ref{stationaryMC1}) and (\ref{stationaryMC2}) is given by
\begin{equation}\label{initialterm}
\pi_{1,\delta_2} = \left\{{
\begin{matrix}
\begin{split}
   &{{{p\left[ {1 - {{\left( {1 - q} \right)}^{\delta_2 - 1}}} \right]} \over q}x, \quad \delta_1 \ge \delta_2-1,} \\
   &{{{p\left[ {1 - {{\left( {1 - q} \right)}^{\delta_1}}} \right]} \over q}x, \quad \delta_1 \le \delta_2-1.}  \\
\end{split}
\end{matrix}
}\right.
\end{equation}
The term $N = {\left( {n - l} \right)\left( {{\delta _1} + 1} \right) + m}+1$ with
\begin{equation}\label{n}
n = \left\lfloor {{{d - {\delta _2}} \over {{\delta _1} + 1}}} \right\rfloor ,
\end{equation}
\begin{equation}\label{m}
m = d - {\delta _2}\bmod {\delta _1} + 1.
\end{equation}
\end{theorem}
In Theorem \ref{TheoremDistribution}, we can see that the stationary distribution of the DTR policy is different for the two cases $\delta_1 \ge \delta_2-1$ and $\delta_1 \le \delta_2-1$ given in (\ref{stationaryMC1}) and (\ref{stationaryMC2}), respectively. This is because that the DTR policy adopts two thresholds $\delta_1$, $\delta_2$, and they can both influence the evolution of the ages at the relay and the destination. Moreover, for each case, we summarize the stationary distribution terms into three subspaces $A$, $B$ and $C$. Subspaces $A$, $B$ and $C$ represent the states with $k \ge 2$, $d = 0$, states with $\forall k$, $2 \le d \le \delta_2$, and states with $\forall k$, $d \ge \delta_2$, respectively. Note that the union space of $A$, $B$ and $C$ forms the entire state space $\mathcal{V}$ given in (\ref{statespace}). The stationary distribution for all the states in subspaces $A$, $B$ and $C$ are summarized in (8a), (8b) and (8c), respectively when $\delta_1 \ge \delta_2-1$, and in (9a), (9b) and (9c), respectively when $\delta_1 \le \delta_2-1$. It is worth emphasizing that due to the complicated evolution for the states within subspace $C$ when $\delta_1 \ge \delta_2-1$, we use an approximation\footnote{We use the approximation that $\pi_{d,0} \approx \left(1-p\right)\pi_{d-1,0}$ for $d \ge \delta_2$, the approximation is very tight for relative large value of $q$ and becomes exact when $q \to 1$. For detailed explanation, please refer to the discussion above (51) in Appendix B.} to obtain (8c) in a closed-form, and as shown in our simulations, the approximation is very tight for relatively large value of $q$. With the help of the stationary distribution for the state space $\mathcal{V}$, we finally characterize the average AoI, and the average number of forwarding actions, for the proposed $\theta_{\rm{DTR}}$ policy given in the following two Theorems.
\begin{theorem}\label{TheoremAoI}
The average AoI of the $\theta_{\rm{DTR}}$ policy is given in (\ref{AoI1}) and (\ref{AoI2}) for the case that $\delta_1 \ge \delta_2-1$, and the case that $\delta_1 \le \delta_2-1$, respectively.
\begin{equation}\label{AoI1}
\begin{split}
\bar \Delta\left(\theta_{\rm{DTR}}\right) &\approx {1 \over p} + {1 \over q} + {\delta _2} - {{\left( {p{\delta _1} - q{\delta _2}} \right){{\left( {1 - q} \right)}^{{\delta _1}}} + q{\delta _2}\left[ {{{p\left( {{\delta _2} - 1} \right)} \over 2} + 1} \right] + 1} \over {q\left( {1 - p} \right) + pq{\delta _2} + p{{\left( {1 - q} \right)}^{{\delta _2} - 1}} - p{{\left( {1 - q} \right)}^{{\delta _1}}}}} +\\
&\quad {{\left[ {1 - {{\left( {1 - q} \right)}^{{\delta _2} - 1}}} \right]\left[ {p\left( {p{\delta _1} - q{\delta _1} - q} \right){{\left( {1 - q} \right)}^{{\delta _1}}} + p - q - pq{\delta _1} + {{q\left( {p{\delta _1} + 1} \right)} \over {1 - {{\left( {1 - q} \right)}^{{\delta _1}}}}}} \right]} \over {p\left[ {1 - {{\left( {1 - q} \right)}^{{\delta _1}}}} \right]\left[ {q\left( {1 - p} \right) + pq{\delta _2} + p{{\left( {1 - q} \right)}^{{\delta _2} - 1}} - p{{\left( {1 - q} \right)}^{{\delta _1}}}} \right]}}, \delta_1 \ge \delta_2-1.
\end{split}
\end{equation}
\begin{equation}\label{AoI2}
\begin{split}
\bar \Delta\left(\theta_{\rm{DTR}}\right)& = {1 \over q} + {1 \over {p\left[ {1 - {{\left( {1 - q} \right)}^{{\delta _1}}}} \right]}} + {{{\delta _1} + {\delta _2}} \over 2} - {{\left[ {{\delta _2} + \left( {1 - p} \right){\delta _1} + p{\delta _1}{\delta _2}} \right]/2} \over {\left( {1 - p} \right) + p{\delta _2} - p\left( {{\delta _2} - {\delta _1} - 1} \right){{\left( {1 - q} \right)}^{{\delta _1}}}}}, \delta_1 \le \delta_2-1.
\end{split}
\end{equation}
\end{theorem}
\begin{remark}
In Theorem \ref{TheoremAoI}, we first can see that (\ref{AoI1}) is an approximate expression while (\ref{AoI2}) is an exact expression. This is due to the approximation of (8c) when $\delta_1 \ge \delta_2-1$. We next investigate the average AoI for the proposed DTR policy for some special cases based on the analytical results derived in (\ref{AoI1}) and (\ref{AoI2}).
When $\delta_1 = \delta_2-1$, the average AoI of the DTR policy can be obtained from both (\ref{AoI1}) and (\ref{AoI2}) given by
\begin{equation}\label{remarkeq1}
\mathop {{{\bar \Delta }\left(\theta_{\rm{DTR}}\right)}}\limits_{{\delta _1} = {\delta _2} - 1} ={1 \over q} + {{{\delta _1} + 1} \over 2} + {1 \over {p\left[ {1 - {{\left( {1 - q} \right)}^{{\delta _1}}}} \right]}} - {{{\delta _1} + 1} \over {2\left( {1 + p{\delta _1}} \right)}}.
\end{equation}
For the case that $q \to 1$, the average AoI of the DTR policy can be simplified to
\begin{equation}\label{remarkeq2}
\mathop {{{\bar \Delta }\left(\theta_{\rm{DTR}}\right)}}\limits_{q \to 1}  = 1 + {1 \over p} + {{p{\delta _2}\left( {{\delta _2} - 1} \right)} \over {2\left( {1 - p + p{\delta _2}} \right)}}.
\end{equation}
For the case that only $\delta_1$ is adopted, i.e., $\delta_2 = 2$, the average AoI can be re-written from (\ref{AoI1}) as
\begin{equation}\label{remarkeq3}
\begin{split}
\mathop {{{\bar \Delta }\left(\theta_{\rm{DTR}}\right)}}\limits_{{\delta _2} = 2} &\approx {1 \over p} + {1 \over q} - {{\left( {p{\delta _1} - 2q} \right){{\left( {1 - q} \right)}^{{\delta _1}}} + pq\left( {{\delta _1} + 1} \right) + 2q + 1} \over {q + p\left[ {1 - {{\left( {1 - q} \right)}^{{\delta _1}}}} \right]}}+ \\
&\quad {{p{\delta _1} + 1} \over {p{{\left[ {1 - {{\left( {1 - q} \right)}^{{\delta _1}}}} \right]}^2}}} - {{q\left[ {p\left( {q{\delta _1} + q} \right){{\left( {1 - q} \right)}^{{\delta _1}}} + q + pq{\delta _1}} \right]} \over {p\left[ {1 - {{\left( {1 - q} \right)}^{{\delta _1}}}} \right]\left[ {q + p\left[ {1 - {{\left( {1 - q} \right)}^{{\delta _1}}}} \right]} \right]}}.
\end{split}
\end{equation}
Note that it is meaningless to set $\delta_2 = 0$ because the age gain $d = 0$ indicates that relay and destination are synchronized. Therefore, the relay should always choose receiving action when $d=0$, and the setting of $\delta_2 = 2$ means that the second threshold $\delta_2$ is not considered in the DTR policy, i.e., $d \ge \delta_2=2$ always satisfies. For the case that only $\delta_2$ is implemented, i.e., $\delta_1 \to \infty$, the average AoI can be simplified from (\ref{AoI1}) and given by
\begin{equation}\label{remarkeq4}
\mathop {{{\bar \Delta }\left(\theta_{\rm{DTR}}\right)}}\limits_{{\delta _1} \to \infty } \approx {1 \over p} + {1 \over q} + {\delta _2} - {{q{\delta _2}\left[ {{{p\left( {{\delta _2} - 1} \right)} \over 2} + 1} \right] + {{\left( {1 - q} \right)}^{{\delta _2} - 1}}} \over {q\left( {1 - p} \right) + pq{\delta _2} + p{{\left( {1 - q} \right)}^{{\delta _2} - 1}}}}.
\end{equation}
We can observe from (\ref{remarkeq1})-(\ref{remarkeq4}) that the increase of $\delta_1$ or $\delta_2$ may not decrease the average AoI of the status update system. For instance, in (\ref{remarkeq1}), on the one hand, the increase of $\delta_1$ can decrease the average AoI because the term ${1 \over {p\left[ {1 - {{\left( {1 - q} \right)}^{{\delta _1}}}} \right]}}$ reduces. On the other hand, it can also increases the average AoI because the term ${{{\delta _1} + 1} \over 2}$ grows. This observation can be explained as follows. In the DTR policy, there exists a tradeoff between receiving action and forwarding action in terms of the average AoI at the destination. Too frequent forwarding action may lead to stale status updates stored at the relay, while too much receiving action results in infrequent reception of status updates at the destination. Therefore, it is critical to tune the two thresholds $\delta_1$ and $\delta_2$ in the proposed DTR policy in order to balance the tradeoff between receiving and forwarding actions to minimize the average AoI at the destination. Similar results can be also found for the general case given in (\ref{AoI1}) and (\ref{AoI2}). Due to the complicated structure of (\ref{AoI1}) and (\ref{AoI2}), it is intractable for us to further determine the optimal values of $\delta_1$ and $\delta_2$ in a closed-form. However, they can be easily solved based on the derived analytical results by numerical methods.
\end{remark}

\begin{theorem}\label{TheoremTP}
The average number of forwarding actions in $\theta_{\rm{DTR}}$ policy when $\delta_1 \ge \delta_2-1$ is given~by
\begin{equation}\label{AVERAGEtp1}
\begin{split}
\eta\left(\theta_{\rm{DTR}}\right) = {{{p\left[ {1 - {{\left( {1 - q} \right)}^{{\delta _1}}}} \right]} \over {\left( {1 - p} \right)q + pq{\delta _2} + p{{\left( {1 - q} \right)}^{{\delta _2} - 1}}\left[ {1 - {{\left( {1 - q} \right)}^{{\delta _1} - {\delta _2} + 1}}} \right]}}}.
\end{split}
\end{equation}
The average number of forwarding actions in $\theta_{\rm{DTR}}$ policy when $\delta_1 \le \delta_2-1$ is given~by
\begin{equation}\label{AVERAGEtp2}
\eta\left(\theta_{\rm{DTR}}\right) ={{{p\left[ {1 - {{\left( {1 - q} \right)}^{{\delta _1}}}} \right]} \over {\left( {1 - p} \right)q + pq{\delta _2} - pq\left( {{\delta _2} - {\delta _1} - 1} \right){{\left( {1 - q} \right)}^{{\delta _1}}}}}}.
\end{equation}
\end{theorem}
\begin{remark}
In Theorem \ref{TheoremTP}, we can deduce from (\ref{AVERAGEtp1}) and (\ref{AVERAGEtp2}) that the decrease of $\delta_1$ and the increase of $\delta_2$ can both reduce the average number of forwarding actions in the proposed $\theta_{\rm{DTR}}$. This is understandable because from (\ref{DTR}), the $\theta_{\rm{DTR}}$ policy tends to choose receiving actions when $\delta_1$ reduces or $\delta_2$ grows. Jointly considering the analytical expressions for the average AoI characterized in Theorem \ref{TheoremAoI}, and the average number of forwarding actions derived in Theorem \ref{TheoremTP}, we can solve the age minimization problem formulated in (\ref{optimizationP}) and find the optimal values of $\delta_1$ and $\delta_2$ in the DTR policy numerically by a two-dimensional exhaustive search.
\end{remark}

\section{Proof of Theorems \ref{theoremCMDP1} and \ref{theoremCMDP2}: Optimal CMDP Policy and its Structure}
In this section, we first prove Theorem \ref{theoremCMDP1} by recasting problem (\ref{optimizationP}) as a CMDP problem. We then validate Theorem \ref{theoremCMDP2} by analyzing the structure of the optimal stationary policy derived from the CMDP model.
\subsection{Proof of Theorem \ref{theoremCMDP1}}
The formulated CMDP model for problem (\ref{optimizationP}) can be described by the 5-tuple $\{\mathcal{V},\mathcal{W},\mathrm{P},r,c\}$. The state space $\mathcal{V}$ and the action space $\mathcal{W}$ have already been defined in Section II. We now give definitions for the other three tuples.
\begin{itemize}
	\item Transition probabilities $\mathrm{P} = \Pr\left\{v^\prime|v,w\right\}$: The probability of transition from state $v$ to $v^\prime$ when taking action $w$. According to the AoI evolution of $R$ and $D$ given in \eqref{AoIR} and \eqref{AoID}, we have all the none-null transition probabilities given below
%	\begin{equation}\label{e3}
%	\begin{aligned}
%	&\Pr\left\{\left(k_t+1,d_t\right)\left| {\left(k_t,d_t\right),\Upsilon_t=0} \right.\right\}=1-p,\\
%		&\Pr\left\{\left(1,k_t+d_t\right)\left| {\left(k_t,d_t\right),\Upsilon_t=0} \right.\right\}=p,\\
%	&\Pr\left\{\left(k_t+1,d_t\right)\left| {\left(k_t,d_t\right),\Upsilon_t=1} \right.\right\}=1-q,\\
%	&\Pr\left\{\left(k_t+1,0\right)\left| {\left(k_t,d_t\right),\Upsilon_t=1} \right.\right\}=q,\\
%	\end{aligned}.
%	\end{equation}
	\begin{equation}\label{e3}
\begin{aligned}
&\Pr\left\{\left(k+1,d\right)\left| {\left(k,d\right),w=0} \right.\right\}=1-p,\\
&\Pr\left\{\left(1,k+d\right)\left| {\left(k,d\right),w=0} \right.\right\}=p,\\
&\Pr\left\{\left(k+1,d\right)\left| {\left(k,d\right),w=1} \right.\right\}=1-q,\\
&\Pr\left\{\left(k+1,0\right)\left| {\left(k,d\right),w=1} \right.\right\}=q.\\
\end{aligned}
\end{equation}
	\item  The reward function $r$: $\mathcal{V} \times \mathcal{W}  \rightarrow \mathbb{R}$ is the AoI at $D$, and it is defined as $r\left(v,w\right)=k+d$.
	\item  The cost function $c$: $\mathcal{V} \times \mathcal{W}  \rightarrow \mathbb{R}$ is the cost for the action $w$, and it is given by $c\left(v,w\right)=w$.
\end{itemize}

We use $J\left( \theta \left| s_0 \right. \right)$ and $C\left( \theta \left| s_0 \right. \right)$ to denote the infinite horizon average AoI and average number of forwarding actions, respectively, when the policy $\theta$ is implemented, given the initial state $s_0$. The problem \eqref{optimizationP} can be transformed into the CMDP problem given by
\begin{equation}\label{cmdp}
	\begin{aligned}
	& \text{Minimize}  \quad J\left( \theta \left| s_0 \right.  \right) \buildrel \Delta \over = \mathop {\lim \sup }\limits_{T \to \infty } {1 \over T}\mathbb{E}\left[ {\sum\limits_{t = 1}^T {\left[ {{a_R}\left( t \right) + g\left( t \right)} \right] \left| s_0 \right.} } \right],\\
	& {\text{Subject to}} \quad  C\left( \theta \left| s_0 \right. \right) \buildrel \Delta \over = \mathop {\lim \sup }\limits_{T \to \infty } {1 \over T}\mathbb{E}\left[ {\sum\limits_{t = 1}^T {w\left( t \right)}\left| s_0 \right. } \right] \le {\eta _C}.
	\end{aligned}
\end{equation}

In order to solve the CMDP problem \eqref{cmdp}, we now apply the Lagrange relaxation with a Lagrange relaxation multiplier $\lambda>0$. The average Lagrangian cost for a given policy $\theta$ and multiplier $\lambda$ is defined as
\begin{equation}\label{mdp}
{L_\lambda }\left( \theta  \right) = \mathop {\lim }\limits_{T \to \infty } {1 \over T}\left( {\mathbb{E}\left[ {\sum\limits_{t = 1}^T {\left[ {{a_R}\left( t \right) + g\left( t \right)} \right]} } \right] + \lambda \mathbb{E}\left[ {\sum\limits_{t = 1}^T {w\left( t \right)} } \right]} \right).
\end{equation}
With the above definitions and according to \cite[Th. 2.5]{sennott1993constrained}, we can now claim the following result:
If there exists a given multiplier ${\lambda}^*>0$ such that $C\left(\theta_{\lambda^*}\left| s_0\right.\right)=\eta_C$, the optimal policy for (\ref{cmdp}) is the optimal policy $\theta_{\lambda^*}$ for the unconstrained problem that minimizes (\ref{mdp}). Otherwise, if no such ${\lambda}^*$ can be found, there exists an optimal stationary policy $\theta_{\rm{CMDP}}$ for problem (\ref{cmdp}) which is a randomized mixture of two stationary deterministic policies $\theta_{\lambda_1^*}$ and $\theta_{\lambda_2^*}$ that differ in at most a single state. Each deterministic policy $\theta_{{\lambda}^*}$, ${\lambda}^* \in \left\{{\lambda_1^*}, {\lambda_2^*}\right\}$ is the optimal solution to the unconstrained problem (\ref{mdp}) for a given multiplier ${\lambda}^*$. The optimal values of the multipliers ${\lambda_1^*}$ and ${\lambda_2^*}$ can be solved by iterative algorithms such as Robbins-Monro algorithm \cite{spall2005introduction}. The policy $\theta_{\rm{CMDP}}$ selects $\theta_{\lambda_1^*}$ with a probability $\alpha$, and chooses $\theta_{\lambda_2^*}$ with a probability $1-\alpha$. The randomization parameter $\alpha\in[0,1]$ is mapped from $\lambda_1^*$ and $\lambda_2^*$, and given by \cite[Eq. 11]{AoI_HARQ}.
\subsection{Proof of Theorem \ref{theoremCMDP2}}
We now turn to Theorem \ref{theoremCMDP2} and analyze the structure of each stationary deterministic policy $\theta_{\lambda^*}$, ${\lambda}^* \in \left\{{\lambda_1^*}, {\lambda_2^*}\right\}$. To proceed, the switching-type policy described in Theorem \ref{theoremCMDP2} is equivalent to the two conclusions: (1) When relay AoI $k$ is fixed, the policy $\theta_{\lambda^*}$ is a monotone nondecreasing policy in terms of $d$; (2) When age gain $d$ is fixed, the policy $\theta_{\lambda^*}$ is a monotone nonincreasing policy in terms of $k$.
\subsubsection{Proof of Conclusion 1}
We first give the following definition.
\begin{definition}\label{defsubadditive}
	(Subadditive \cite{puterman2014markov}) A multivariable function $Q(v,w): \mathcal{V} \times \mathcal{W} \rightarrow \mathbb{R}$ is subadditive in $(v,w)$ , if for all $v^{+}\geq v^{-}$ and $w^{+}\geq w^{-}$, the following inequality holds
	\begin{equation}
	\label{e8}
	\begin{aligned}
	Q(v^{+},w^{+})+ Q(v^{-},w^{-}) \leq Q(v^{+},w^{-})+ Q(v^{-},w^{+}).
	\end{aligned}
	\end{equation}
\end{definition}
To use the above definition, we order the state by the age gain $d$, i.e., $v^+ > v^-$ if $d^+>d^-$, where $v^+ = (\cdot, d^+)$ and  $v^- = (\cdot, d^-)$. We now give the one-step Lagrangian cost function based on (\ref{mdp}), and it is given by
\begin{equation}\label{onestep}
{L}\left( v,w,\lambda  \right) = k+d+ \lambda w.
\end{equation}
With Definition \ref{defsubadditive}, and function (\ref{onestep}), according to \cite[Th. 8.11.3]{puterman2014markov}, Conclusion 1 holds if the following four conditions hold.
\begin{itemize}
	\item [(a)] ${L}\left( v,w,\lambda  \right)$ is nondecreasing in $v$ for all $w \in \mathcal{W}$;
    \item [(b)] ${L}\left( v,w,\lambda  \right)$ is a subadditive function on $\mathcal{V}\times \mathcal{W}$
	\item [(c)] $\Phi(u|v,w)=\sum_{v^\prime=u}^{\infty}\Pr(v^\prime|v,w)$ is nondecreasing in $v$ for all $u\in \mathcal{V}$ and $w\in \mathcal{W}$ where $\Pr(v^\prime|v,w)$ is the transition probability given in \eqref{e3};
	\item [(d)] $\Phi(u|v,w)=\sum_{v^\prime=u}^{\infty}\Pr(v^\prime|v,w)$ is a subadditive function on $\mathcal{V}\times \mathcal{W}$ for all $u\in \mathcal{V}$.
\end{itemize}
It is straightforward to prove that Conditions (a) and (b) hold through some mathematical manipulations. We now turn to Conditions (c) and (d), according to the transition probabilities given in \eqref{e3}, if the current state is $v=(k,d)$, the next possible states are $v_1=(1,k+d)$, $v_2=(k+1,d)$ and $v_3=(k+1,0)$, we have
\begin{equation}
\label{req2}
\Phi(u|v,w=0)=\left\{
\begin{array}{rcl}
1,& &\text{if } u\leq v_2,  \\
1-p,& &\text{if } v_2 <u \le v_1, \\
0,& &\text{if }  u > v_1, \\
\end{array}
\right.
\end{equation}
\begin{equation}
\label{req4}
\Phi(u|v,w=1)=\left\{
\begin{array}{rcl}
1,& & \text{if } u \le v_3, \\
1-q,& &\text{if } v_3 < u \le v_2, \\
0,& &\text{if }  u > v_2.\\
\end{array}
\right.
\end{equation}
With the results for function $\Phi(u|v,w)$ derived in (\ref{req2}) and (\ref{req4}), we can then prove Condition (c) holds. At last, we turn to Condition (d), from (\ref{e8}), we realize that there are three possible combinations of $w^{+}, w^{-}$ in the considered problem: (i) $w^+=1$, $w^-=0$, (ii) $w^+=w^-=1$ and (iii) $w^+=w^-=0$. By considering those three combinations, and substitute them into the expression of $\Phi(u|v,w)$ derived in (\ref{req2}) and (\ref{req4}), we can finally verify that Condition (d) holds. Till this end, by verifying Conditions (a)-(d), we have shown that Conclusion 1 hold.
\subsubsection{Proof of Conclusion 2}
In Conclusion 2, $d$ is fixed and $\theta_{\lambda^*}$ is a monotone nonincreasing policy in terms of $k$. Due to the limited space, we will not give the detailed proof for Conclusion 2 and it can be verified with a similar procedure as Conclusion~1.

%Fig. \ref{fig:cmdppolicy} illustrates the structure of the optimal policy for the relaxed unconstrained MDP. State truncation has been applied to deal with countable but infinite state space. Specifically, we truncate the instantaneous AoI of both $R$ and $D$ to a maximum value $a_{max}=50$ in getting Fig. \ref{fig:cmdppolicy}, thus, the plotted result shows a triangle shape.
%\begin{figure}
%\centering \scalebox{0.4}{\includegraphics{cmdpolicy.eps}}
%\caption{Structural deterministic policy for $\lambda_1$ (top) and $\lambda_2$ (bottom) where $\lambda_1>\lambda_2$ with $p=0.6$, $q=0.7$ and $\eta_c=0.2$.}\label{fig:cmdppolicy}
%\end{figure}
%\subsubsection{Threshold-based relaying (PTR)}
%The performance of the proposed TR scheme can be further increased by introducing a probability for choosing ${\Upsilon _D}$ and ${\Upsilon _R}$. Specifically, in the proposed PTR scheme, when the age of the stored status update is equal or below $\delta$, the relay chooses ${\Upsilon _D}$ with probability $p$, and ${\Upsilon _F}$ with probability $1-p$. When the age of the stored status update exceeds $\delta$, the relay chooses ${\Upsilon _D}$ as in the TR scheme.

%\input{performance_analysis2}
%\input{performance_analysis3}
\section{Proof of Theorem \ref{TheoremDistribution}: Stationary Age Distribution of DTR Policy}
We proof Theorem \ref{TheoremDistribution} by first modelling the status update system implementing the DTR policy as an MC, and then derive the stationary distribution of the MC.
\subsection{MC Modelling}
The state of the MC has already been defined in Section II-C, e.g., $v\triangleq \left(k,d\right)$, with a state space $\mathcal{V}$ given in (\ref{statespace}). We now define the transition probability $\Pr\left\{v^\prime \left| {v} \right.\right\}$ for the MC as the probability of transition from state $v$ to $v^\prime$. Substitute the actions taken by the $\theta_{\rm{DTR}}$ policy described in (\ref{DTR}) into the transition probability given in (\ref{e3}), we have all the none null transition probabilities for the MC given by
\begin{subequations}\label{MCtransition}
\begin{align}
&\Pr\left\{\left({k+1,d}\right)\left| \left({{k,d}}\right) \right.\right\}=1-p, \quad k > \delta_1 \cup d < \delta_2,\\
&\Pr\left\{\left({k+1,d}\right)\left| \left({{k,d}}\right) \right.\right\}=1-q,\quad k \le \delta_1 \cap d \ge \delta_2,\\
&\Pr\left\{\left({k+1,0}\right)\left| \left({{k,d}}\right) \right.\right\}=q,\quad k \le \delta_1 \cap d \ge \delta_2,\\
&\Pr\left\{\left({1,k+d}\right)\left| \left({{k,d}}\right) \right.\right\}=p,\quad k > \delta_1 \cup d < \delta_2.
\end{align}
\end{subequations}
For a better understanding, we depict the transition of the $\theta_{\rm{DTR}}$ policy in Fig. \ref{fig:markov1} for (29a) and (29b). We show the transition for (29c) in Fig. \ref{fig:markov2}. We use two Figs \ref{fig:markov3} and \ref{fig:markov4} to illustrate the transition for (29d). Specifically, Fig. \ref{fig:markov3} depicts the transition (29d) for the states with $\forall k, d \le \delta_2$, and Fig. \ref{fig:markov4} considers the transition (29d) for the states with $\forall k, d \ge \delta_2$. Figs. \ref{fig:markov1}-\ref{fig:markov4} will be used throughout this section to explain the derived analytical results for the purposes of comprehensive illustration.
\begin{figure}
\centering \scalebox{0.35}{\includegraphics{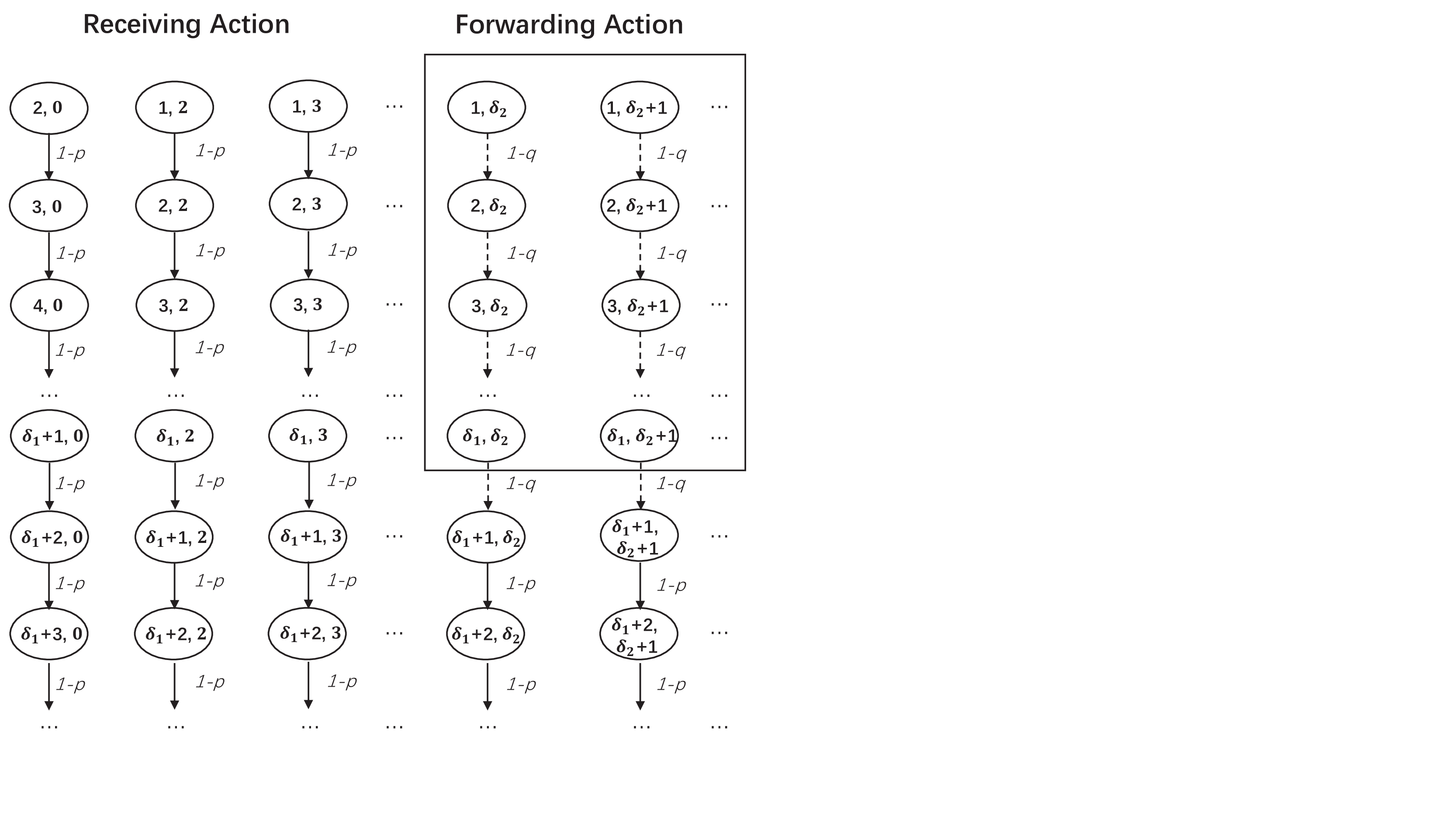}}
\caption{State transition of the proposed DTR policy for (29a) and (29b).}\label{fig:markov1}
\end{figure}
\begin{figure}
\centering \scalebox{0.35}{\includegraphics{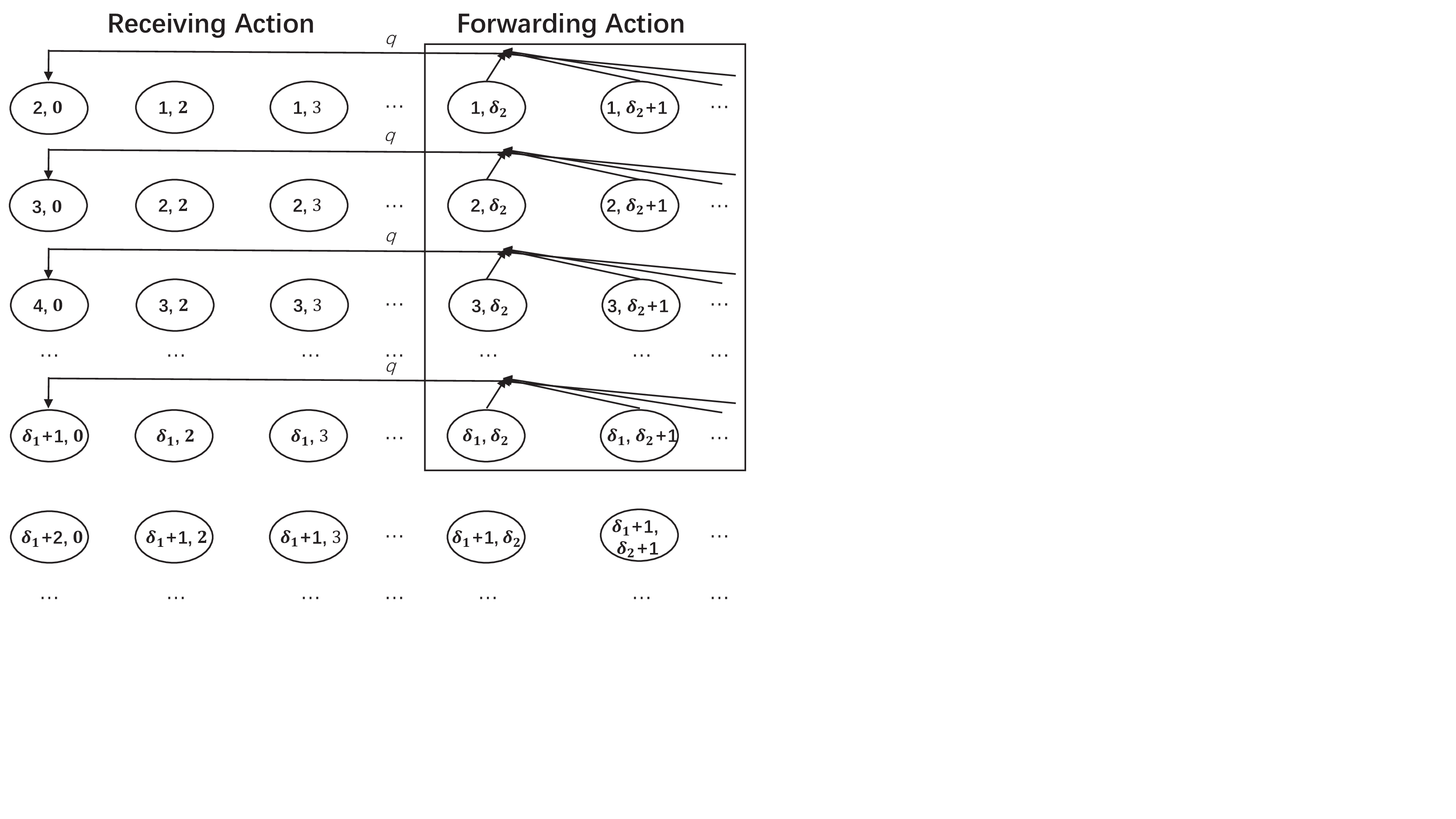}}
\caption{State transition of the proposed DTR policy for (29c).}\label{fig:markov2}
\end{figure}
\begin{figure}
\centering \scalebox{0.35}{\includegraphics{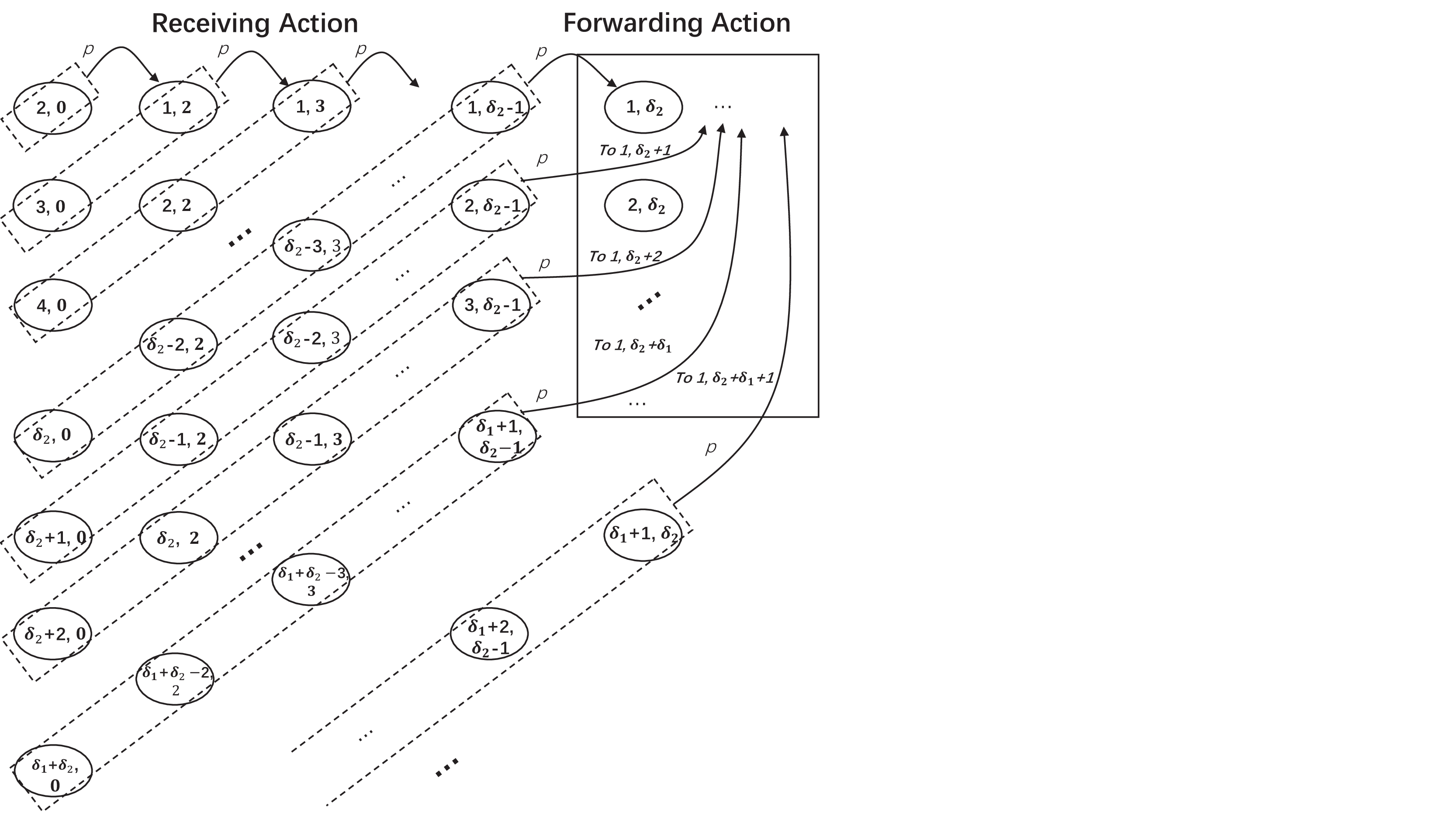}}
\caption{State transition of the proposed DTR policy for (29d) for states with $\forall k, d \le \delta_2$.}\label{fig:markov3}
\end{figure}
\begin{figure}
\centering \scalebox{0.35}{\includegraphics{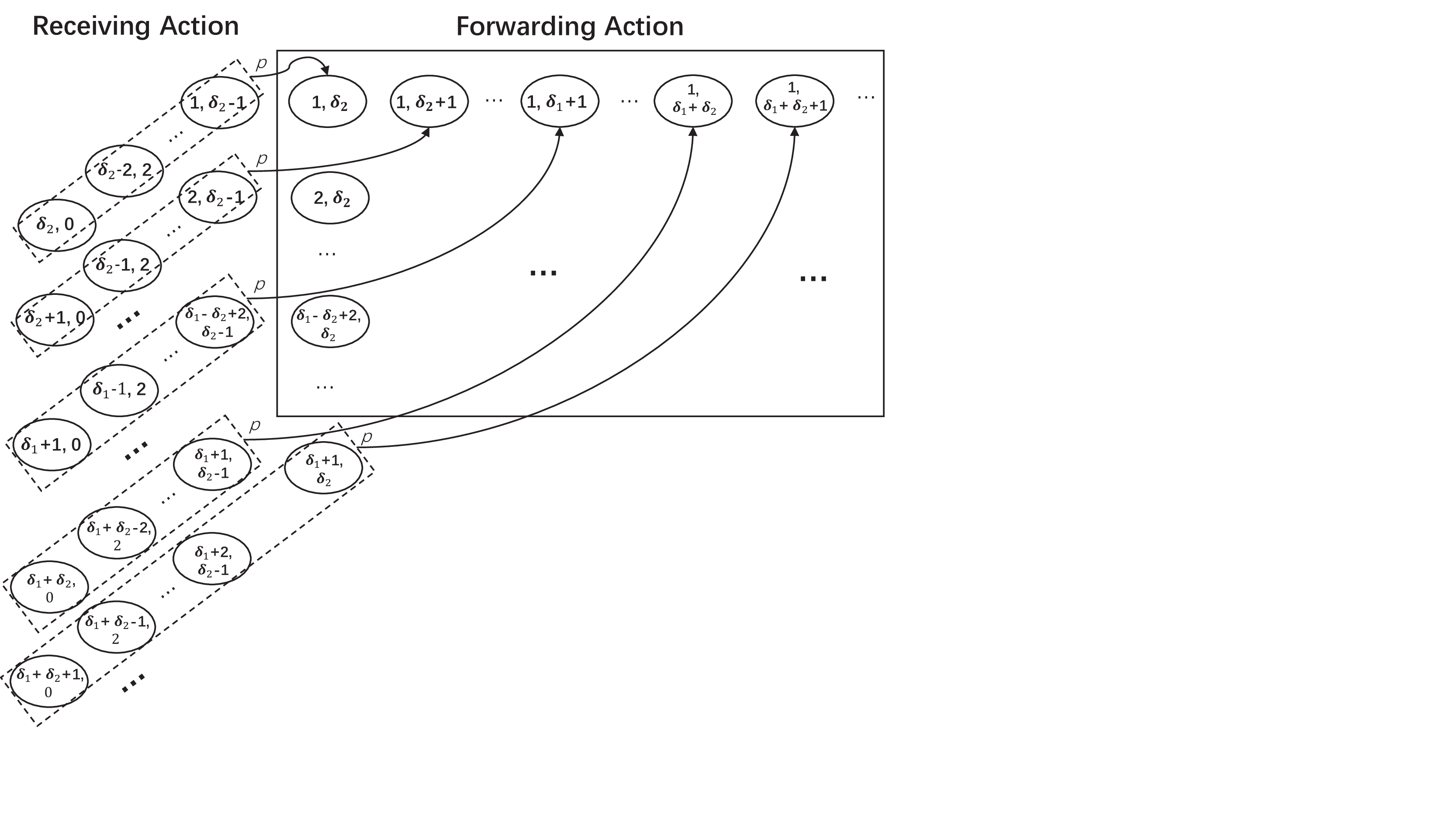}}
\caption{State transition of the proposed DTR policy for (29d) for states with $\forall k, d \ge \delta_2$.}\label{fig:markov4}
\end{figure}
\subsection{Stationary Distribution}
\begin{figure}
\centering \scalebox{0.35}{\includegraphics{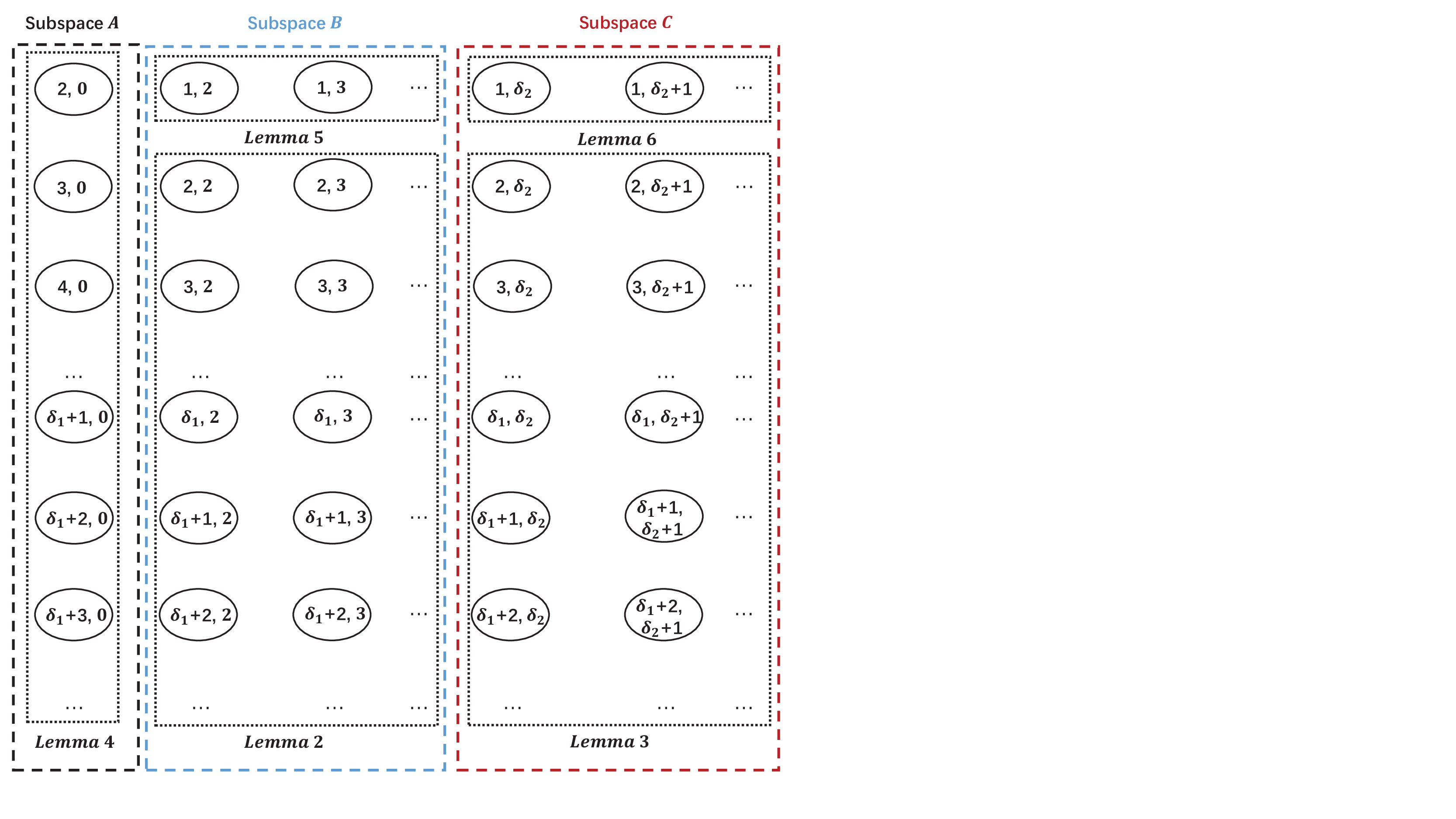}}
\caption{The mappings of Lemmas \ref{property1}-\ref{property5} to the subspaces $A$, $B$ and $C$.}\label{fig:explain}
\end{figure}
With the transition probabilities of the MC, we then derive the stationary distribution of the MC and proof Theorem \ref{TheoremDistribution}. We use Lemmas \ref{property1}-\ref{property5} to prove the analytical expressions of the stationary distribution given in (\ref{stationaryMC1}) and (\ref{stationaryMC2}). Recall that we summarize the stationary distribution in (\ref{stationaryMC1}) and (\ref{stationaryMC2}) by considering three subspaces $A$, $B$ and $C$. For a better understanding, we depict the mappings of the Lemmas to the subspaces $A$, $B$ and $C$ in Fig. \ref{fig:explain}. Specifically, Lemmas \ref{property1} and \ref{property2} depict the evolution of the stationary distribution for the states in subspaces $B$ and $C$ in terms of each column, respectively. In other words, if the stationary distribution of the first row in subspaces $B$ and $C$ is obtained, i.e., $\pi_{1,d}$, $d \ge 2$, by using Lemmas \ref{property1} and \ref{property2}, the stationary distribution for all the other states in subspaces $B$ and $C$ can be derived. Lemma \ref{property3} shows the evolution of the stationary distribution for the states in subspace $A$, which can be used to solve the stationary distribution for states in subspace $A$. Lemmas \ref{property4} and \ref{property5} present the evolution of the stationary distribution for the first row of subspaces $B$ and $C$, respectively. They can be used to derive the stationary distribution for the states in the first row of subspaces $B$ and $C$, respectively. We provide the proofs for all the Lemmas given in this subsection in Appendix B. With the above explanations, we are ready to give the Lemmas.
\begin{lemma}\label{property1}
The evolution of $\pi_{k,d}$ for each column in subspace $B$, e.g., $\forall k, 2\le d \le \delta_2-1$, is given~by
\begin{equation}\label{fact1}
\pi_{k,d} = \left(1-p\right)\pi_{k-1,d} =  \left(1-p\right)^{k-1}\pi_{1,d}, \forall k, 2\le d \le \delta_2-1.
\end{equation}
\end{lemma}
\begin{lemma}\label{property2}
The evolution of $\pi_{k,d}$ for each column in subspace $C$, e.g., $\forall k, d \ge \delta_2$, is given by
\begin{equation}\label{fact2}
\pi_{k,d} =
\left\{{
\begin{matrix}
\begin{split}
   &{\left( {1 - q} \right)^{k-1}{\pi _{1,d}},1 \le k \le \delta_1+1,d \ge \delta_2 }, \\
   &{\left( {1 - q} \right)^{\delta_1}\left(1-p\right)^{k-\delta_1-1}{\pi _{1,d}}, k \ge \delta_1+1,d \ge \delta_2}.  \\
\end{split}
\end{matrix}
}\right.
\end{equation}
\end{lemma}
Lemmas \ref{property1} and \ref{property2} can be used to derive the stationary distribution for all the states within subspaces $B$ and $C$ once the stationary distribution of the first row is obtained, i.e., $\pi_{1,d}$. Let $\pi_{2,0} = x$, we then present Lemma \ref{property3}.
\begin{lemma}\label{property3}
The evolution of $\pi_{k,d}$ in subspace $A$, e.g., $\forall k$ and $d = 0$, is given by
\begin{equation}\label{fact3}
{\pi _{k,0}} =
\left\{{
\begin{matrix}
\begin{split}
   &{ {\pi _{k - 1,0}}\left( {1 - p} \right) + x{\left( {1 - q} \right)^{k - 1}},3 \le k \le {\delta _1} + 1 }, \\
   &{ {\pi _{k - 1,0}}\left( {1 - p} \right) = {\pi _{{\delta _1} + 1,0}}{\left( {1 - p} \right)^{k - {\delta _1} - 1}},k \ge {\delta _1} + 1},  \\
\end{split}
\end{matrix}
}\right.
\end{equation}
where the initial term of the evolution is $\pi_{2,0}$, i.e., $x$.
\end{lemma}
We can solve the stationary distribution for the states in subspace $A$ based on the evolution described in Lemma \ref{property3} and the initial term $x$. We use a mathematical induction method to obtain the stationary distribution expressions for subspace $A$ given in (8a) and (9a). Due to the limited space, we will not provide the detailed analysis of the mathematical induction method. The analytical expressions in (8a) and (9a) can be easily verified by substituting them into (\ref{fact3}).
\begin{lemma}\label{property4}
The evolution of $\pi_{k,d}$ for the first row in subspace $B$, e.g., $k=1$ and $2 \le d \le \delta_2$ is given below. When $\delta_1 \ge \delta_2-1$,
\begin{equation}\label{fact41}
\pi_{1,d} ={\pi _{1,d-1}} + p{\left( {1 - q} \right)^{d - 2}}x, 3\le d \le \delta_2,\\
\end{equation}
when $\delta_1 \le \delta_2-1$,
\begin{equation}\label{fact42}
\pi_{1,d} =
\left\{{
\begin{matrix}
\begin{split}
   &{{\pi _{1,d-1}} + p{{\left( {1 - q} \right)}^{d - 2}}x,3 \le d \le {\delta _1+1} }, \\
   &{{\pi _{1,d-1}}, {\delta _1}+2  \le d \le {\delta _2}},  \\
\end{split}
\end{matrix}
}\right.
\end{equation}
where the initial term of the evolution $\pi_{1,2} = px$.
\end{lemma}
Lemma \ref{property4} can be used to solve the stationary distribution for the first row of subspace $B$ by using a mathematical induction method. By jointly considering Lemma \ref{property1}, we can characterize the stationary distribution for all the states in subspace $B$ in (8b) and (9b) for the two cases when $\delta_1 \ge \delta_2-1$, and $\delta_1 \le \delta_2-1$, respectively. Similarly, the analytical results derived in (8b) and (9b) can be verified by substituting them into (\ref{fact41}) and (\ref{fact42}), respectively.
\begin{lemma}\label{property5}
The evolution of $\pi_{k,d}$ for the first row in subspace $C$, e.g., $k=1$ and $d \ge \delta_2$ is given~by
\begin{equation}\label{fact5}
\begin{split}
\pi_{k,d} \approx \left\{{
\begin{matrix}
\begin{split}
   &{\left( {1 - p} \right){\pi _{1,d-1}}, \delta_2+1 \le d \le {\delta _1+\delta_2} } \\
   &{\left[\left( {1 - p} \right){\pi _{1,d-1}} + p\left( {1 - q} \right)^{\delta_1}{\pi _{1,d - {\delta _1-1}}}\right], d \ge {\delta _1+\delta_2+1}}   \\
\end{split}
\end{matrix}
}\right.,
\end{split}
\end{equation}
where the initial term of the evolution $\pi_{1,\delta_2}$ is given by (\ref{initialterm}).
\end{lemma}
Due to the complicated evolution of the states in subspace $C$ when $\delta_1 \ge \delta_2-1$, we use an approximation in deriving Lemma \ref{property5} for the first row in subspace $C$. Similar to Lemma \ref{property4}, we can characterize (8c) and (9c) from (\ref{fact5}). Note that (8c) is an approximate expression due to the adopted approximation. We have now characterized the stationary distribution of all the states in a closed-form in (8) and (9) based on $x$, i.e. $\pi_{2,0}$. At last, we show the following Lemma to solve $x$ in a closed-form.
\begin{lemma}\label{lemmax}
The term $x$, i.e., $\pi_{2,0}$ is given in (\ref{x}).
\end{lemma}
We obtain Lemma \ref{lemmax} by using the fact that the summation of the stationary distribution for all the states is equal to $1$. It is worth emphasizing that we obtain the summation of the stationary distribution without using Lemma \ref{property5} where an approximation is adopted. Therefore, the expression of $x$ given in (\ref{x}) is exact. Till this end, we have derived the stationary distribution in a closed-form and finished the proof of Theorem \ref{TheoremDistribution}. Note that the proofs for all the Lemmas in this section are given in Appendix B.

\section{Proof of Theorems \ref{TheoremAoI} and \ref{TheoremTP}: Performance Analysis for DTR}
In this section, by using the age stationary distribution of the $\theta_{\rm{DTR}}$ policy presented in Theorem \ref{TheoremDistribution}, we calculate the average AoI and the average number of forwarding actions for the proposed $\theta_{\rm{DTR}}$ policy. All the proofs of the Lemmas in this section are given in Appendix C.
\subsection{Proof of Theorem \ref{TheoremAoI}}
Recall that each state $\left(k,d\right)$ is defined as the instantaneous AoI of $R$ being $k$, and the AoI gain between $D$ and $R$ being $d$. The average AoI of the considered cooperative system using the $\theta_{\rm{DTR}}$ policy can thus be expressed as
\begin{equation}\label{DTRexpression}
\begin{split}
\bar \Delta \left(\theta_{\rm{DTR}}\right)&=\sum\limits_{\forall k,d} {\left( {k + d} \right){\pi _{k,d}}} \\
&=\sum\limits_{k = 2}^\infty  {k{\pi _{k,0}}}  + \sum\limits_{k = 1}^\infty  {\sum\limits_{d = 2}^{{\delta _2} - 1} {\left( {k + d} \right){\pi _{k,d}}} } + \sum\limits_{k = 1}^\infty  {\sum\limits_{d = {\delta _2}}^\infty  {\left( {k + d} \right){\pi _{k,d}}} }. \\
\end{split}
\end{equation}
In the following, we characterize the three terms $\sum\limits_{k = 2}^\infty  {k{\pi _{k,0}}}$, $\sum\limits_{k = 1}^\infty  {\sum\limits_{d = 2}^{{\delta _2} - 1} {\left( {k + d} \right){\pi _{k,d}}} }$ and $\sum\limits_{k = 1}^\infty  {\sum\limits_{d = {\delta _2}}^\infty  {\left( {k + d} \right){\pi _{k,d}}} }$ in (\ref{DTRexpression}) by the Lemmas \ref{AoItermlemma}-\ref{AoIterm3lemma} and they represent the AoI terms for subspaces $A$, $B$ and $C$, respectively.
\begin{lemma}\label{AoItermlemma}
The term $\sum\limits_{k = 2}^\infty  {k{\pi _{k,0}}}$ in (\ref{DTRexpression}) for subspace $A$ is given by
\begin{equation}\label{AoIterm0}
\sum\limits_{k = 2}^\infty  {k{\pi _{k,0}}} ={{\left( {p + q} \right)\left[ {1 - {{\left( {1 - q} \right)}^{{\delta _1}}}} \right] - pq{\delta _1}{{\left( {1 - q} \right)}^{{\delta _1}}}} \over {{p^2}{q^2}}}x.
\end{equation}
\end{lemma}
\begin{lemma}\label{AoIterm2lemma}
The term $\sum\limits_{k = 1}^\infty  {\sum\limits_{d = 2}^{{\delta _2} - 1} {\left( {k + d} \right){\pi _{k,d}}} }$ in (\ref{DTRexpression}) for subspace $B$ is given in (\ref{AoIterm2case1}) and (\ref{AoIterm2case2}) for the two cases that $\delta_1 \ge \delta_2-1$, and $\delta_1 \le \delta_2-1$, respectively.
\begin{equation}\label{AoIterm2case1}
\begin{split}
\sum\limits_{k = 1}^\infty  {\sum\limits_{d = 2}^{{\delta _2} - 1} {\left( {k + d} \right){\pi _{k,d}}} }& = {{q{\delta _2} - q - \left[ {1 - {{\left( {1 - q} \right)}^{{\delta _2} - 1}}} \right]} \over {p{q^2}}}x+{{\left( {{\delta _2} + 1} \right)\left( {{\delta _2} - 2} \right)} \over {2q}}x \\
&\quad - {{1 - {q^2} - \left( {1 - q + q{\delta _2}} \right){{\left( {1 - q} \right)}^{{\delta _2} - 1}}} \over {{q^3}}}x, \delta_1 \ge \delta_2-1,
\end{split}
\end{equation}
\begin{equation}\label{AoIterm2case2}
\begin{split}
&\sum\limits_{k = 1}^\infty  {\sum\limits_{d = 2}^{{\delta _2} - 1} {\left( {k + d} \right){\pi _{k,d}}} } = {{q{\delta _2} - q - 1 + \left( {1 - q{\delta _2} + q{\delta _1} + q} \right){{\left( {1 - q} \right)}^{{\delta _1}}}} \over {{pq^2}}}x+{{\left( {{\delta _2} + 1} \right)\left( {{\delta _2} - 2} \right)} \over {2q}}x - \\
&\quad {{1 - {q^2} - \left( {1 + q + q{\delta _1}} \right){{\left( {1 - q} \right)}^{{\delta _1} + 1}}} \over {{q^3}}}x - {{{{\left( {1 - q} \right)}^{{\delta _1}}}\left( {{\delta _1} + {\delta _2} + 1} \right)\left( {{\delta _2} - {\delta _1} - 2} \right)} \over {2q}}x, \delta_1 \le \delta_2-1.
\end{split}
\end{equation}
\end{lemma}
\begin{lemma}\label{AoIterm3lemma}
The term $\sum\limits_{k = 1}^\infty  {\sum\limits_{d = {\delta _2}}^\infty  {\left( {k + d} \right){\pi _{k,d}}} }$ in (\ref{DTRexpression}) for subspace $C$ is by
\begin{equation}\label{AoIterm3case1}
\begin{split}
&\sum\limits_{k = 1}^\infty  {\sum\limits_{d = {\delta _2}}^\infty  {\left( {k + d} \right){\pi _{k,d}}} }\\
&\approx {{{{\left( {1 - q} \right)}^{{\delta _1}}}\left( {{{{\delta _1} + {\delta _2}} \over p} + {1 \over {{p^2}}} - {{{\delta _1} + {\delta _2}} \over q} - {1 \over {{q^2}}}} \right) + {{{\delta _2}} \over q} + {1 \over {{q^2}}}} \over q}x +\left[ {\left( {{1 \over p} - {1 \over q}} \right){{\left( {1 - q} \right)}^{{\delta _1}}} + {1 \over q}} \right]x \times\\
&\quad \left\{ {{{{\delta _1}{{\left( {1 - q} \right)}^{{\delta _1}}}\left[ {1 - {{\left( {1 - q} \right)}^{{\delta _2} - 1}}} \right]} \over {q{{\left[ {1 - {{\left( {1 - q} \right)}^{{\delta _1}}}} \right]}^2}}} + {{\left[ {1 - {{\left( {1 - q} \right)}^{{\delta _2} - 1}}} \right]} \over {pq{{\left[ {1 - {{\left( {1 - q} \right)}^{{\delta _1}}}} \right]}^2}}} - {{\left[ {1 - {{\left( {1 - q} \right)}^{{\delta _2} - 1}}} \right]} \over {q\left[ {1 - {{\left( {1 - q} \right)}^{{\delta _1}}}} \right]}}} \right\}, \delta_1 \ge \delta_2-1,\\
\end{split}
\end{equation}
\begin{equation}\label{AoIterm3case2}
\begin{split}
&\sum\limits_{k = 1}^\infty  {\sum\limits_{d = {\delta _2}}^\infty  {\left( {k + d} \right){\pi _{k,d}}} }\\
&={{{{\left( {1 - q} \right)}^{{\delta _1}}}\left( {{{{\delta _1} + {\delta _2}} \over p} + {1 \over {{p^2}}} - {{{\delta _1} + {\delta _2}} \over q} - {1 \over {{q^2}}}} \right) + {{{\delta _2}} \over q} + {1 \over {{q^2}}}} \over q}x +\\
&\quad \left\{ {{{{\delta _1}{{\left( {1 - q} \right)}^{{\delta _1}}}} \over {q\left[ {1 - {{\left( {1 - q} \right)}^{{\delta _1}}}} \right]}} + {1 \over {pq\left[ {1 - {{\left( {1 - q} \right)}^{{\delta _1}}}} \right]}} - {1 \over q}} \right\}\left[ {\left( {{1 \over p} - {1 \over q}} \right){{\left( {1 - q} \right)}^{{\delta _1}}} + {1 \over q}} \right]x, \delta_1 \le \delta_2-1.\\
\end{split}
\end{equation}
\end{lemma}
By using Lemmas \ref{AoItermlemma}-\ref{AoIterm3lemma} and (\ref{DTRexpression}), after some mathematical manipulations, we can obtain the desired results given in (\ref{AoI1}) and (\ref{AoI2}) for the two cases $\delta_1 \ge \delta_2-1$ and $\delta_1 \le \delta_2-1$, respectively.
\subsection{Proof of Theorem \ref{TheoremTP}}
According to the definition of the average number of forwarding actions given in (\ref{longtermenergy}), and the $\theta_{\rm{DTR}}$ policy described in (\ref{DTR}), the average number of forwarding actions can be evaluated as
\begin{equation}
\eta\left(\theta_{\rm{DTR}}\right) = \sum\limits_{k = 1}^{{\delta _1}} {\sum\limits_{d = {\delta _2}}^\infty  {{\pi _{k,d}}} }.
\end{equation}
Recall that $x = \pi_{2,0}$, from the fact that state $\left({2,0}\right)$ can only be reached from states $\left({1,d}\right)$, $d \ge \delta_2$ with a probability of $q$, we have $x = \sum\limits_{d = {\delta _2}}^\infty  {{\pi _{1,d}}}q$ and thus $\sum\limits_{d = {\delta _2}}^\infty  {{\pi _{k-1,d}}} = {x \over q}$. By considering the property given in Lemma \ref{property2}, we can obtain that $\sum\limits_{d = {\delta _2}}^\infty  {{\pi _{k,d}}} = \sum\limits_{d = {\delta _2}}^\infty  {{\pi _{1,d}}} \left(1-q\right)^{k-1}= {x \over q} \left(1-q\right)^{k-1}, 1 \le k \le \delta_1+1$. With this result, we can further simplify the above expression as
\begin{equation}\label{TPeq2}
\eta\left(\theta_{\rm{DTR}}\right)= \sum\limits_{k = 1}^{{\delta _1}} {{x \over q}{{\left( {1 - q} \right)}^{k - 1}}}  = {{1 - {{\left( {1 - q} \right)}^{{\delta _1}}}} \over {{q^2}}}x.
\end{equation}
Substitute $x$ given in (\ref{x}) into (\ref{TPeq2}), we have characterized the desired result in Theorem \ref{TheoremTP}.

\section{Numerical and Simulation Results}
In this section, we present the numerical and simulation results of the considered two-hop status update system applying the proposed CMDP-based and DTR policies. We first verify Theorems \ref{theoremCMDP1} and \ref{theoremCMDP2} by depicting the optimal policy $\theta_{\rm{CMDP}}$ for different resource constraints $\eta_C$ in Fig. \ref{FIGCMDP} when $p=0.6$ and $q=0.7$. Without loss of generality, we set $s_0 = \left(2,0\right)$ as the initial state of the CMDP problem in the following simulations. The optimal policies in Fig 6 for different system setups are obtained by Robbins-Monro algorithm \cite{spall2005introduction} for the search of $\lambda$ and Relative Value Iteration (RVI) for the corresponding policy. We apply RVI on finite states by setting the maximum instantaneous AoI for both relay and destination to 200 in order to approximate the countable infinite state space \cite{stochasticbook}. From Figs. \ref{fig:cmdp1} and \ref{fig:cmdp2}, we can see that $\theta_{\rm{CMDP}}$ is a randomized mixture of two deterministic policies $\theta_{\lambda_1^*}$ and $\theta_{\lambda_2^*}$, and they differs at a single state $\left(1,4\right)$ in Fig \ref{fig:cmdp1}, and state $\left(4,3\right)$ in Fig. \ref{fig:cmdp2}, respectively. In Fig. \ref{fig:cmdp3}, the two deterministic policies $\theta_{\lambda_1^*}$ and $\theta_{\lambda_2^*}$ are identical for the case that $\eta_C = 0.65$. All of these observations coincide well with Theorem \ref{theoremCMDP1}. Furthermore, the depicted policies in Fig. \ref{FIGCMDP} follows the switching-type structure and have multiple thresholds, which verifies our theoretical analysis provided in Theorem \ref{theoremCMDP2}. At last, as the resource constraint at the relay becomes loose, the optimal policy $\theta_{\rm{CMDP}}$ tends to select the forwarding actions more frequently. This is understandable as the instantaneous AoI at the destination can be potentially reduced by the forwarding actions.
\begin{figure}
\centering
 \subfigure[Resource Constraint $\eta_C = 0.25$]
  {\scalebox{0.36}{\includegraphics {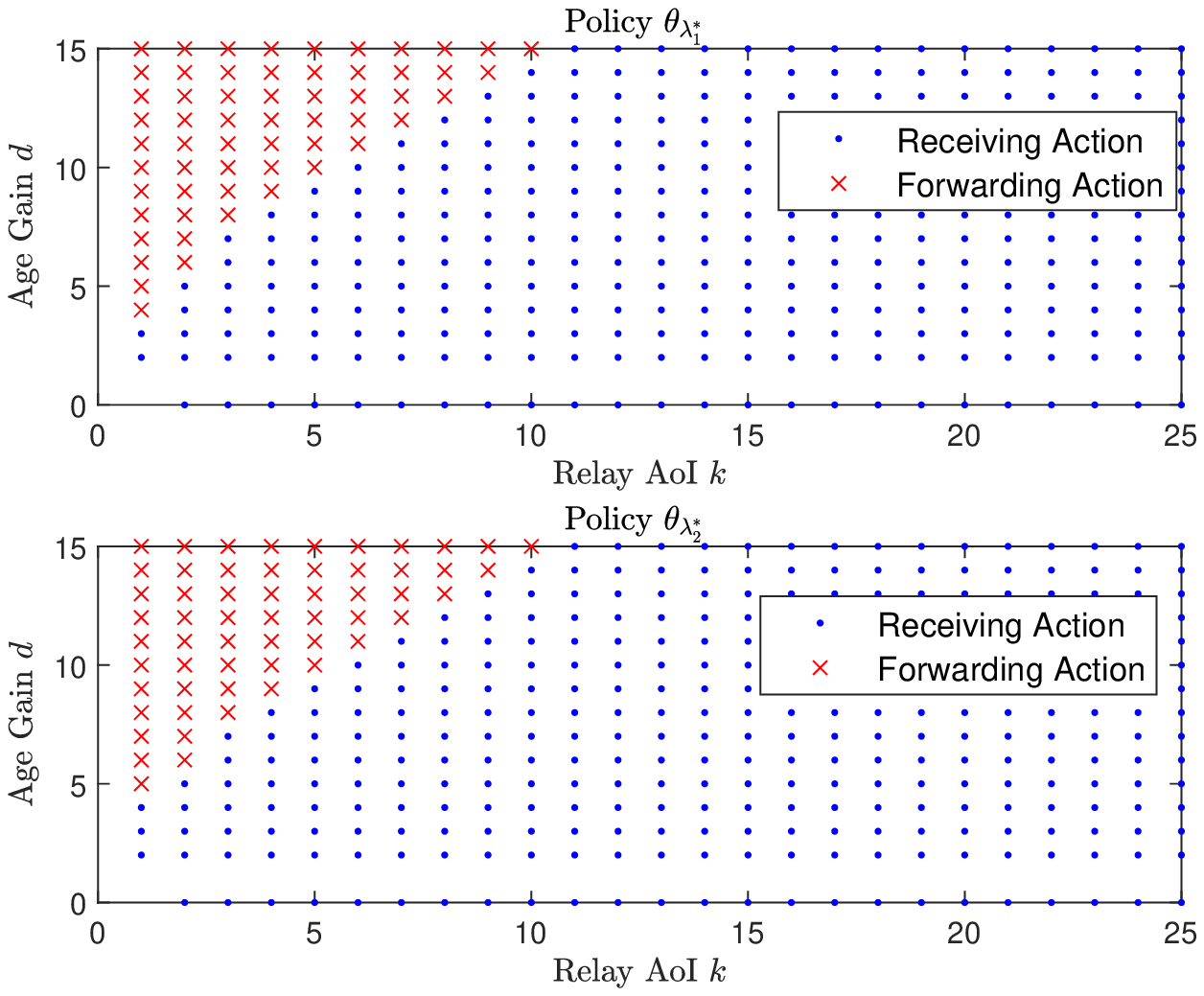}
  \label{fig:cmdp1}}}
\hfil
 \subfigure[Resource Constraint $\eta_C = 0.45$]
  {\scalebox{0.36}{\includegraphics {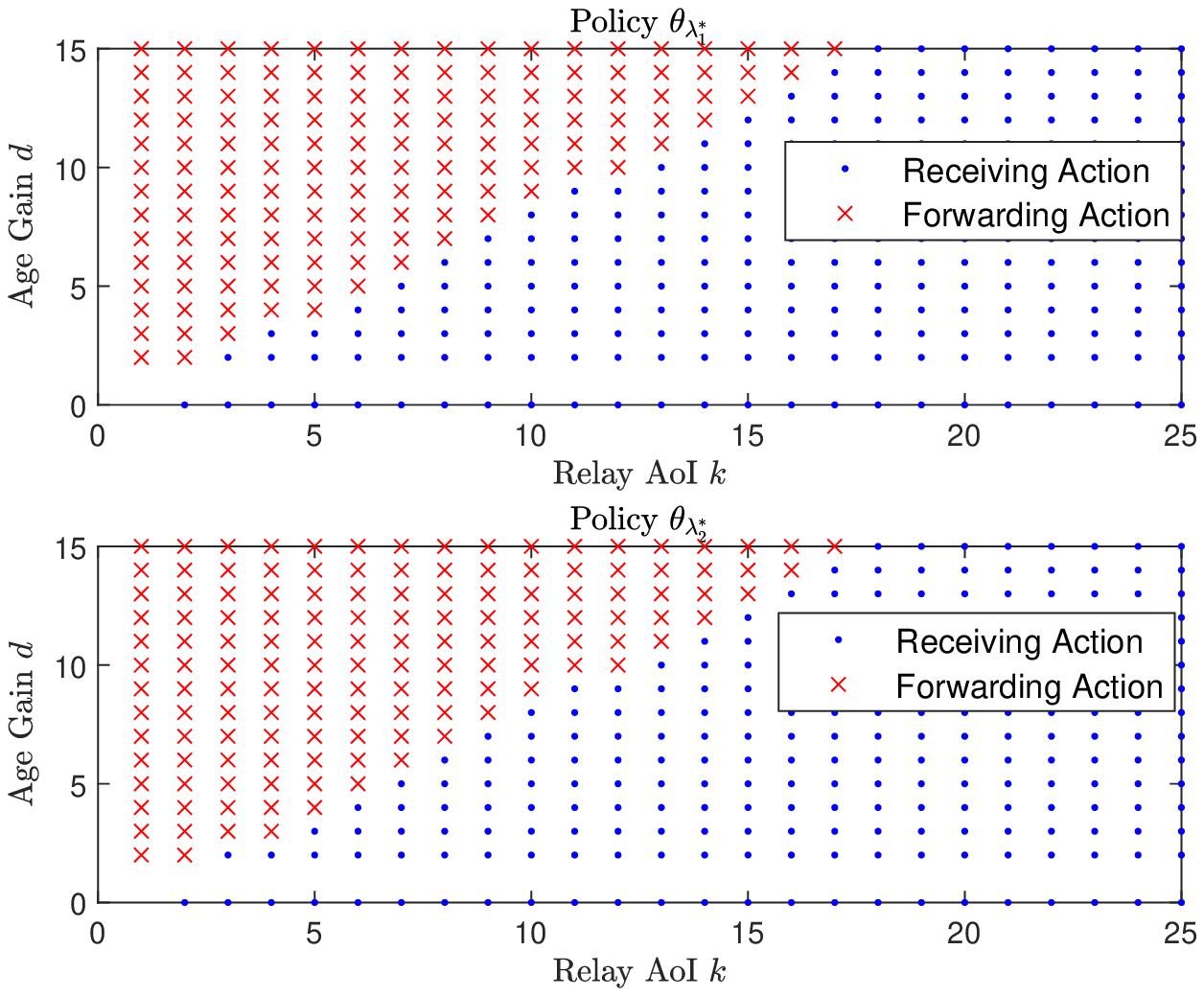}
\label{fig:cmdp2}}}
\hfil
 \subfigure[Resource Constraint $\eta_C =0.65$]
  {\scalebox{0.36}{\includegraphics {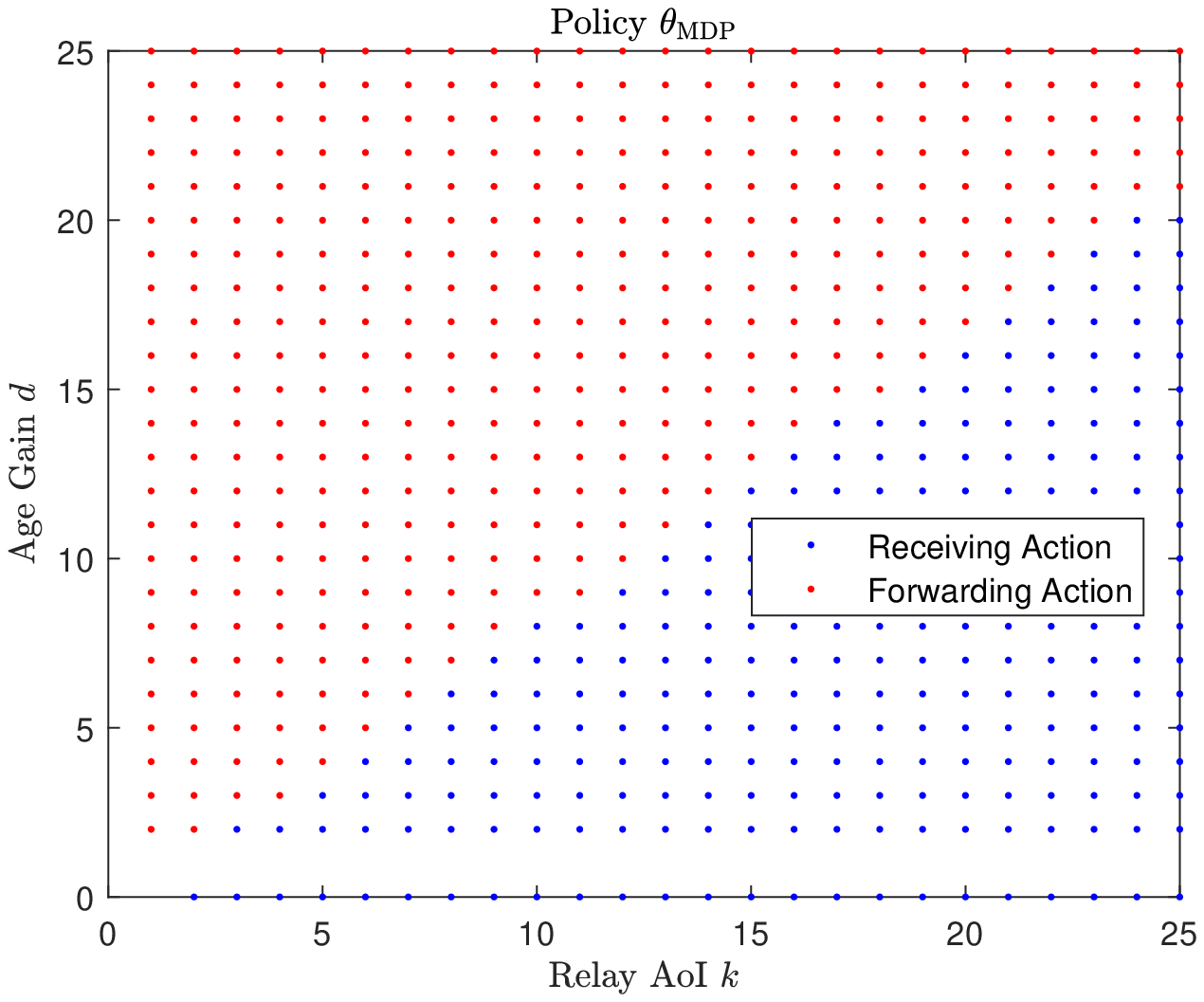}
\label{fig:cmdp3}}}
\caption{Optimal deterministic policies $\theta_{\lambda_1^*}$ and $\theta_{\lambda_2^*}$ for the $\theta_{\rm{CMDP}}$ for different resource constraints $\eta_C$ where $p = 0.6$ and $q = 0.7$.
\label{FIGCMDP}}
\end{figure}

We next plot the average AoI, and the average number of forwarding actions for the DTR policy in Figs. \ref{analytical1} and \ref{analytical2} for different combinations of $p, q, \delta_1$ and $\delta_2$. The analytical results shown in Figs. \ref{analytical1} and \ref{analytical2} are based on (\ref{AoI1}), (\ref{AoI2}) for the average AoI, and (\ref{AVERAGEtp1}), (\ref{AVERAGEtp2}) for the average number of forwarding actions, respectively. Recall that the analytical expressions are different and depending on the two cases $\delta_1 \ge \delta_2-1$ and $\delta_1 \le \delta_2-1$. We can first observe from Figs. \ref{analytical1} and \ref{analytical2} that the analytical results coincide well with the simulation results which verifies our theoretical analysis given in Theorems \ref{TheoremDistribution}-\ref{TheoremTP}. There is a slightly mismatch in Figs. \ref{analytical11} and \ref{analytical21} is due to our approximation in deriving (\ref{AoI1}) and the approximation is very tight in the simulated cases. Besides, we can conclude from Figs. \ref{analytical1} and \ref{analytical2} that the average AoI of the two-hop system implementing DTR policy reduces as $\delta_1$ increases and $\delta_2$ decreases in the simulated cases. However, the average number of forwarding actions grows when $\delta_1$ increases and $\delta_2$ decreases in the simulated cases. Therefore, there exists a tradeoff between receiving and forwarding actions in terms of the average AoI and the average number of forwarding actions. In order to achieve minimum average AoI under a resource constraint, the thresholds $\delta_1$ and $\delta_2$ need to be tuned to balance the choice between receiving and forwarding actions. This observation agrees with the analysis provided in Remarks 1 and 2.
\begin{figure}
\centering
 \subfigure[Average AoI]
  {\scalebox{0.36}{\includegraphics {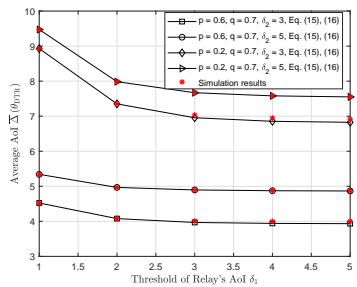}
  \label{analytical11}}}
\hfil
 \subfigure[Average Number of Forwarding Actions]
  {\scalebox{0.36}{\includegraphics {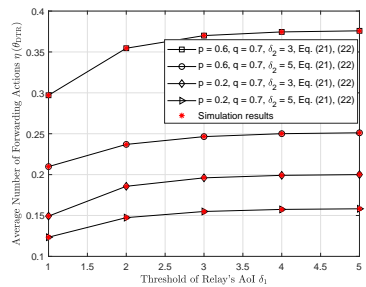}
\label{analytical12}}}
\caption{The average AoI and average number of forwarding actions versus the threshold of relay's AoI $\delta_1$ in the DTR policy for different $p$, $q$ and $\delta_2$.
\label{analytical1}}
\end{figure}
\begin{figure}
\centering
 \subfigure[Average AoI]
  {\scalebox{0.36}{\includegraphics {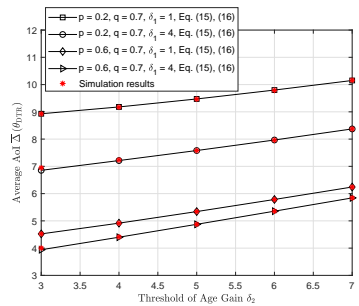}
  \label{analytical21}}}
\hfil
 \subfigure[Average Number of Forwarding Actions]
  {\scalebox{0.36}{\includegraphics {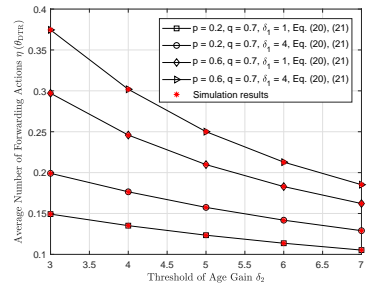}
\label{analytical22}}}
\caption{The average AoI and average number of forwarding actions versus the threshold of age gain $\delta_2$ in the DTR policy for different $p$, $q$ and $\delta_1$.
\label{analytical2}}
\end{figure}

We now show the structure of our proposed DTR policy for different resource constraint in Fig. \ref{FIGDTR}, where the two thresholds $\delta_1$ and $\delta_2$ are optimized numerically based on the closed-form expressions (\ref{AoI1}), (\ref{AoI2}) for the average AoI, and (\ref{AVERAGEtp1}), (\ref{AVERAGEtp2}) for the average number of forwarding actions. We can now compare the structure of the optimal CMDP-based policy with the proposed sub-optimal DTR policy under the same system settings by referring to Figs. \ref{FIGCMDP} and \ref{FIGDTR}. We can observe that the proposed DTR policy takes more forwarding actions when the instantaneous age at the relay is relatively low, at the cost of taking no forwarding actions when the instantaneous age at the relay is relatively high. Besides, the structure of the proposed DTR policy is simple and easy to be implemented. At last, we can deduce from Fig. \ref{FIGDTR} that when the resource constraint is relatively loose, e.g., $\eta_C = 0.45, 0.65$, only one threshold of relay's AoI $\delta_1$ is needed in the proposed DTR policy and the threshold of the age gain $\delta_2 = 2$. Recall that in Remark 1, we have discussed that $\delta_2=2$ indicates the threshold on the age gain is not considered in the DTR policy. This observation can be explained as follows. First of all, the relay should always forward status updates with low instantaneous age and receive new status updates when the stored one becomes stale to keep the age at the destination as low as possible. Therefore, the threshold $\delta_1$ is needed for all the simulated cases. However, when the resource constraint is tight, e.g., $\eta_C = 0.25$, we need a second threshold $\delta_2$ to further balance the receiving and forwarding actions of the relay. Together with $\delta_1$, the relay only forwards those status updates with low instantaneous age, and can decrease the instantaneous age at the destination significantly.
\begin{figure}
\centering
 \subfigure[Resource Constraint $\eta_C = 0.25$]
  {\scalebox{0.36}{\includegraphics {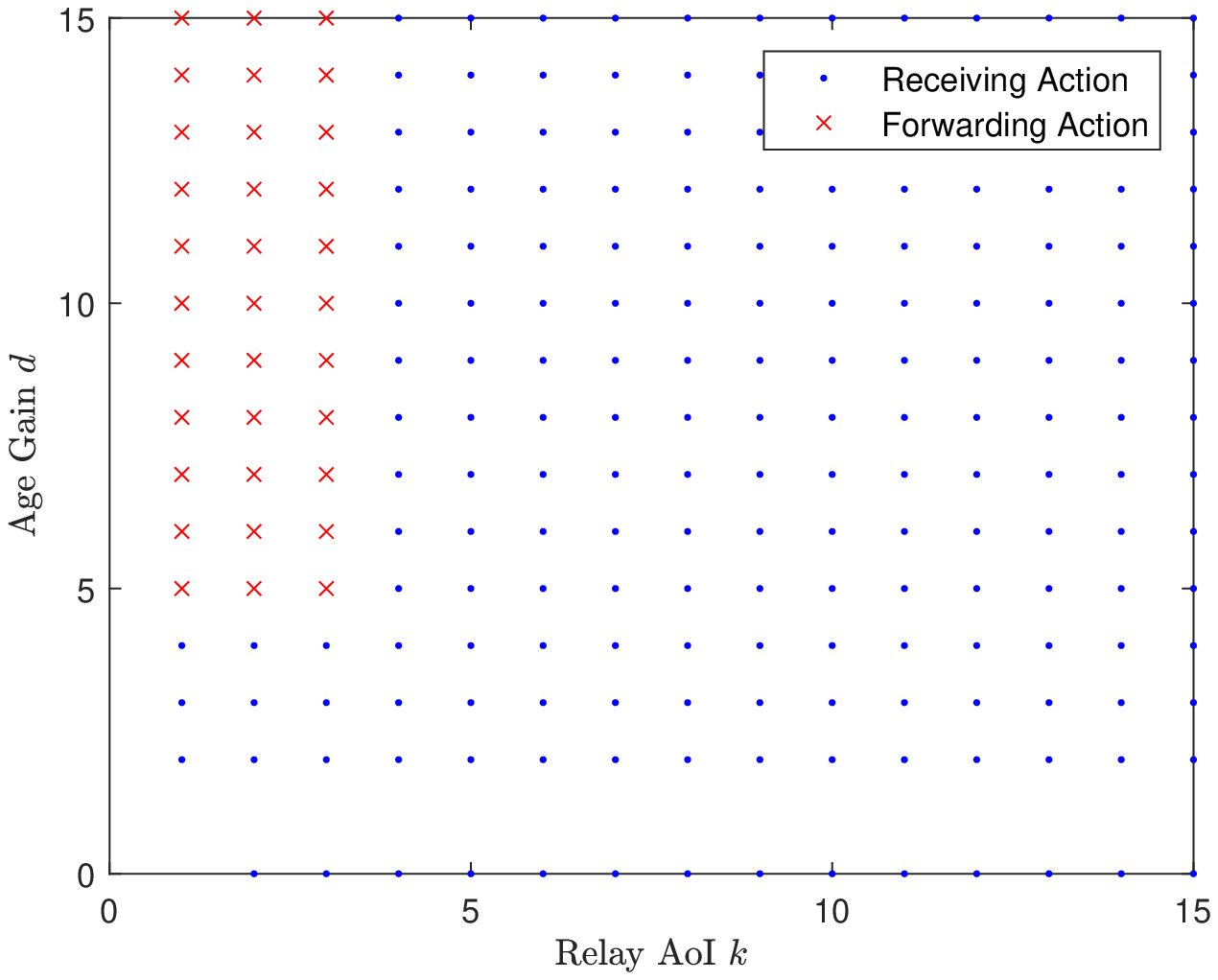}
  \label{fig:DTR1}}}
\hfil
 \subfigure[Resource Constraint $\eta_C = 0.45$]
  {\scalebox{0.36}{\includegraphics {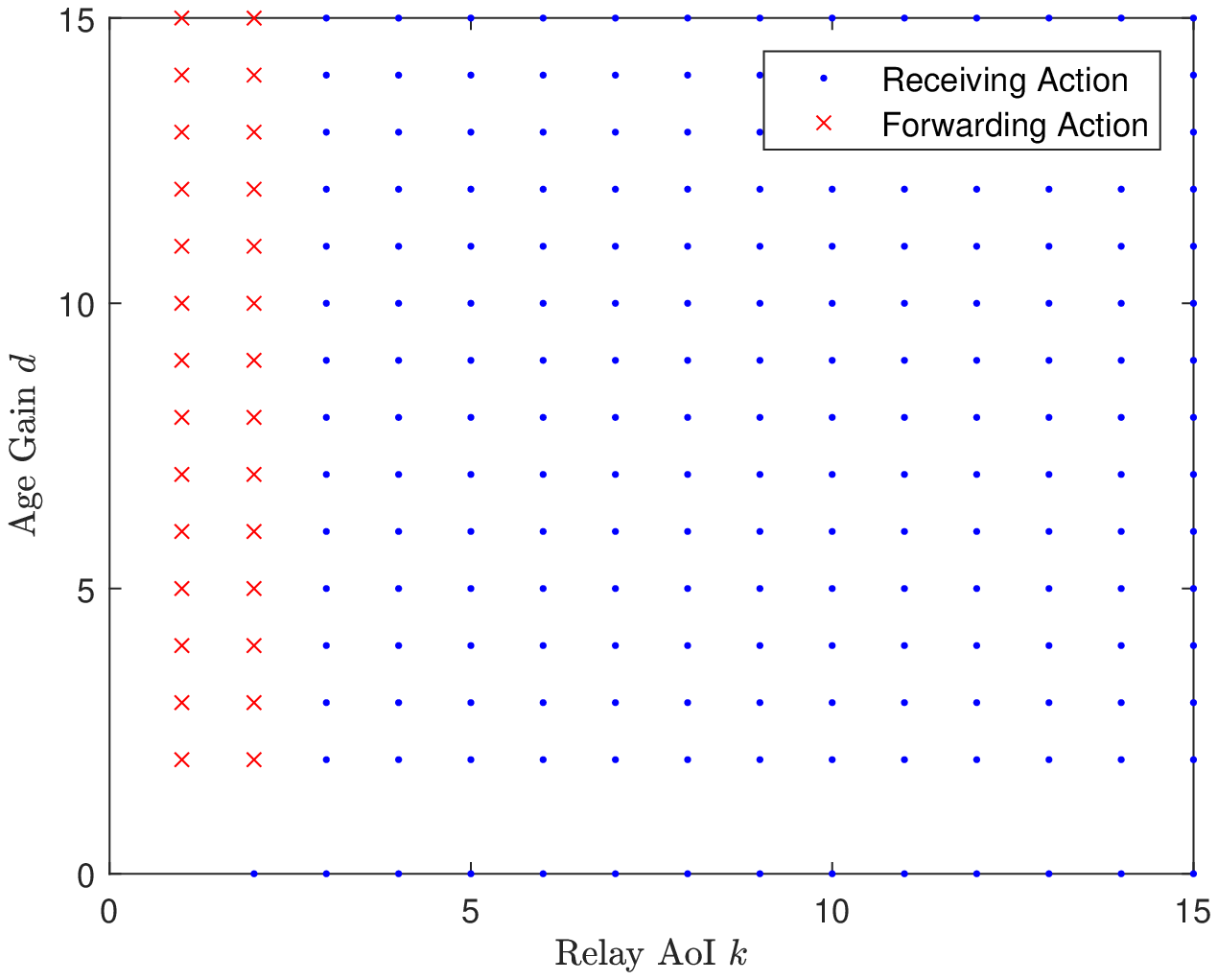}
\label{fig:DTR2}}}
\hfil
 \subfigure[Resource Constraint $\eta_C = 0.65$]
  {\scalebox{0.36}{\includegraphics {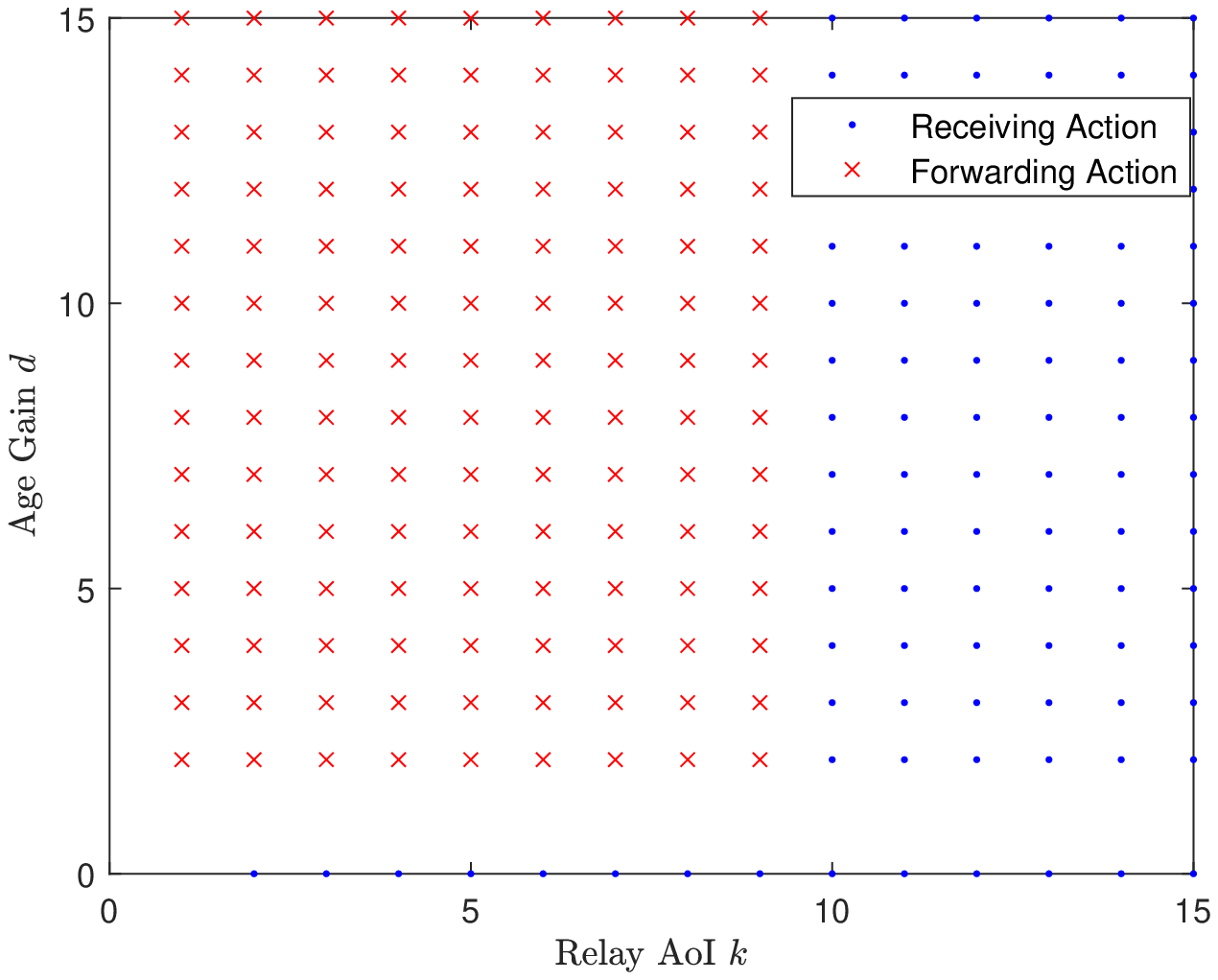}
\label{fig:DTR3}}}
\caption{The proposed DTR policy $\theta_{\rm{DTR}}$ for different resource constraints $\eta_C$ where $p=0.6$ and $q=0.7$.
\label{FIGDTR}}
\end{figure}

Fig. \ref{comparison} compares the average AoI of the proposed CMDP-based policy and DTR policy for different system setups. We simulate the performance of the CMDP-based policy by $1000$ runs, and the time horizon $T$ in each run is set to $10^7$ time slots. In the DTR policy, the minimized average AoI is evaluated by firstly optimizing the two thresholds $\delta_1$ and $\delta_2$ using the derived analytical expressions. In all the simulated cases, we can observe that the proposed low-complexity policy DTR is very close to the optimal policy $\theta_{\rm{CMDP}}$. Therefore, the DTR policy is appealing to practical systems because it can achieve a near-optimal performance and have a simple implementation. At last, the derived closed-form expressions of the DTR policy can further benefit the system design to reveal its performance in terms of the average AoI.
\begin{figure}
\centering \scalebox{0.36}{\includegraphics{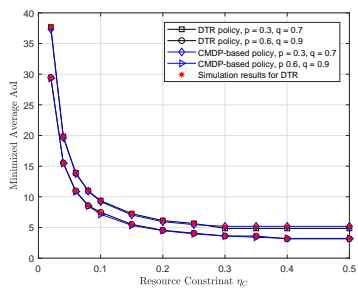}}
\caption{The comparison between CMDP-based and DTR policies in terms of the optimal average AoI for different system setups.}\label{comparison}
\end{figure}

% or

\section{Conclusions}
In this paper, we studied the optimal scheduling policy for a two-hop system where a resource constraint is considered at the relay. In Theorem 1, we first derived an optimal CMDP-based policy to minimize the average AoI at the destination under the resource constraint. The CMDP-based policy was obtained by modelling the optimization problem as a CMDP problem. In Theorem 2, we analytically showed that the CMDP-based policy has a multiple threshold structure with a switching-type. Based on the structure of the CMDP-based policy, we then proposed a low-complexity DTR policy where only two thresholds are implemented, one for the relay's AoI and the other one for the age gain between destination and relay. In Theorem 3, we obtained the stationary distribution of the DTR policy in terms of the instantaneous age at the relay and destination by a MC-based method. In Theorem 4 and 5, the average AoI and the average number of forwarding actions for the DTR policy were characterized in closed-form. Numerical results were provided to verify the Theorems 1-5, and revealed that the proposed DTR policy can achieve near-optimal performance in terms of the average AoI compared with the CMDP-based policy. Furthermore, only one threshold for the relay's AoI is needed in the DTR policy when there is no resource constraint or the resource constraint is loose.
\appendices
\section{Proof of Lemma \ref{lemmastatespace}}\label{app1}
To begin, we have $k\ge 1, d \ge 0$ because that the instantaneous AoI at $R$ is at least one for the $S-R$ link, and the age gain cannot be negative. We use $v = \left(k,d\right)$, $v^\prime = \left(k^\prime,d^\prime\right)$ to be the system states in the current time slot, and in the next time slot, respectively, with $v,v^\prime  \in \mathcal{V}$. With the definitions of $v$ and $v^\prime$, we then obtain the state space $\mathcal{V}$ by considering the evolutions described in (\ref{AoID}) for the following three cases.

Case 1: $R$ chooses forwarding action and successfully delivers a status update to $D$, i.e., $w\left(t\right)=1 \cap I_{RD}\left(t\right) = 1$. We have $k^\prime = k+1 \ge 2$ and $d^\prime = 0$ for this case. Note that $d^\prime = 0$ only happens for this case when $D$ receives a status update correctly, and we thus have the subset $\left(k \ge 2, d = 0\right)$ for the state space $\mathcal{V}$.

Case 2: $R$ chooses receiving action and successfully decodes a status update from $S$, i.e., $w\left(t\right)=0 \cap I_{SR}\left(t\right) = 1$. We have $k^\prime = 1$, and $d^\prime = k + d$. From Case 1, because the minimum value of $k$ is $2$ when $d = 0$, jointly considering the fact that $k \ge 1$, $d \ge 0$, we have $d^\prime \ge 2$, and thus the subset $\left(k =1, d \ge 2 \right)$ for the state space.

Case 3: $R$ neither successfully forwards a status update nor decodes a status update. We have $k^\prime = k+1$, and $d^\prime = d$. In the evolution of Case 3, $k$ increases while $d$ remains unchanged, i.e., $k$ can evolve to any positive integer for a given $d$. Jointly considering the subsets characterized in Cases 1 and 2, we can obtain the state space $\mathcal{V}$ given in Lemma \ref{lemmastatespace}.
\section{Proof of Lemmas in Section IV}\label{AppsectionIV}
\subsection{Proof of Lemma \ref{property1}}
From (\ref{MCtransition}), we can see that state $\left(k,d\right)$ can be only reached by state $\left({k-1,d}\right)$, $\forall k, 2\le d \le \delta_2-1$ with a probability of $1-p$. We thus have the evolution described in this Lemma and such result can also be seen from Fig. \ref{fig:markov1} clearly.
\subsection{Proof of Lemma \ref{property2}}
Similar to the proof given for Lemma \ref{property1}, by using (\ref{MCtransition}) and Fig. \ref{fig:markov1}, we can obtain the evolution described in Lemma \ref{property2}.
\subsection{Proof of Lemma \ref{property3}}
To begin, from (\ref{MCtransition}), we can deduce that state $\left({k,0}\right)$ can be reached from state $\left({k-1,0}\right)$, $3 \le k \le \delta_1+1$ with a probability of $1-p$. Besides, state $\left({k,0}\right)$ can also be reached from states $\left({k-1,d}\right)$, $3 \le k \le \delta_1+1$, $d \ge \delta_2$ with a probability of $q$. The above deduction can also be seen clearly from Figs. \ref{fig:markov1} and \ref{fig:markov2}. With this fact, we can obtain that
\begin{equation}\label{col01equation}
\pi_{k,0} = \pi_{k-1,0}\left(1-p\right)+\sum\limits_{d = {\delta _2}}^\infty  {{\pi _{k-1,d}}}q, 3 \le k \le \delta_1+1.
\end{equation}
Recall that $x = \pi_{2,0}$, from the fact that state $\left({2,0}\right)$ can only be reached from states $\left({1,d}\right)$, $d \ge \delta_2$ with a probability of $q$, we have $x = \sum\limits_{d = {\delta _2}}^\infty  {{\pi _{1,d}}}q$ and thus $\sum\limits_{d = {\delta _2}}^\infty  {{\pi _{k-1,d}}} = {x \over q}$. By considering the property given in Lemma \ref{property2}, we can obtain that
\begin{equation}\label{app1result1}
\sum\limits_{d = {\delta _2}}^\infty  {{\pi _{k,d}}} = \sum\limits_{d = {\delta _2}}^\infty  {{\pi _{1,d}}} \left(1-q\right)^{k-1}= {x \over q} \left(1-q\right)^{k-1}, 1 \le k \le \delta_1+1.
\end{equation}
Substitute (\ref{app1result1}) into (\ref{col01equation}), we have obtained the desired result given in Lemma \ref{property3} for the case that $1 \le k \le \delta_1+1$.

We now turn to the case with $k \ge \delta_1+1$. Differently, the state $\left({k,0}\right)$ can only be reached from one state $\left({k-1,0}\right)$ with a probability of $1-p$, see Figs. \ref{fig:markov1} and \ref{fig:markov2} for instance. Similar to (\ref{col01equation}), we have the desired result given in Lemma~\ref{property3}.
\subsection{Proof of Lemma \ref{property4}}
From the transition probability derived in (\ref{MCtransition}), we can first deduce that state $\left({1,2}\right)$ can only be reached by state $\left({2,0}\right)$ with a probability of $p$, see Fig. \ref{fig:markov3} for instance. Thus, we have $\pi_{1,2} = p{\pi_{2,0}}=px$. Besides, based on (\ref{MCtransition}), state $\pi_{1,d}$ can be reached by states $\pi_{k^\prime,d^\prime}$ with a probability of $p$, when $3 \le d \le \delta_2$ and $\forall k^\prime, d^\prime$ with $k^\prime+d^\prime = d$, see the dashed rectangles in Fig. \ref{fig:markov3} for instance. As a result, we can obtain
\begin{equation}\label{area2case1eq1}
\begin{split}
\pi_{1,d}& =  p\left[ {{\pi _{d,0}} + \sum\limits_{n = 2}^{d-1} {{\pi _{d - n,n}}} } \right]\mathop { = }\limits^{\left( a \right)}p\left[ {{\pi _{d,0}} + \sum\limits_{n = 2}^{d-1} {{\pi _{1,n}}} {{\left( {1 - p} \right)}^{d- n-1}}} \right], 3 \le d \le \delta_2,\\
\end{split}
\end{equation}
where the equality $a$ is due to the property given in Lemma \ref{property1}. We then use Lemma \ref{property3} for the term ${\pi _{d,0}}$ in order to further simplify (\ref{area2case1eq1}). Note that the evolution of ${\pi _{d,0}}$ given in (\ref{fact3}) is different for the two cases $ 3 \le d \le \delta_1+1$, and $ d \ge \delta_1+1$, respectively. As a result, when $\delta_1 \ge \delta_2-1$, we have $d \le \delta_2 \le \delta_1+1$ and only the first case holds. Therefore, substitute (\ref{fact3}) into (\ref{area2case1eq1}), we can further derive that
\begin{equation}\label{area2case1eq11}
\begin{split}
\pi_{1,d} &=  p\left( {1 - p} \right)\left[ {{\pi _{d-1,0}} + \sum\limits_{n = 2}^{d - 2} {{\pi _{1,n}}} {{\left( {1 - p} \right)}^{d - n - 2}}} \right] + p{\left( {1 - q} \right)^{d - 2}}x + p{\pi _{1,d-1}}, 3\le d \le \delta_2\\
& ={\pi _{1,d-1}} + p{\left( {1 - q} \right)^{d - 2}}x, 3\le d \le \delta_2.\\
\end{split}
\end{equation}

We next turn to the case where $\delta_1 \le \delta_2-1$, similarly, by substituting (\ref{fact3}) into (\ref{area2case1eq1}) and consider both cases in (\ref{fact3}), we can obtain that
\begin{equation}\label{area2case2eq2}
\pi_{1,d} =
\left\{{
\begin{matrix}
\begin{split}
   &{{\pi _{1,d-1}} + p{{\left( {1 - q} \right)}^{d - 2}}x,\quad 3 \le d \le {\delta _1+1} } \\
   &{{\pi _{1,d-1}},\quad {\delta _1}+2  \le d \le {\delta _2}}  \\
\end{split}
\end{matrix}
}\right..
\end{equation}
\subsection{Proof of Lemma \ref{property5}}
We first focus on the initial term $\pi_{1, \delta_2}$, substituting $k=1$ and $d = \delta_2$ into (8b) and (9b), we can derive the desired result given in (\ref{initialterm}). Note that (8b) and (9b) are derived based on Lemmas \ref{property2} and \ref{property4}. With the initial term, we then consider the evolution for $\pi_{1,d}$ with $d \ge \delta_2+1$. We realize that the evolution of the term $\pi_{1,d}$ is different for the case that $\delta_2+1 \le d \le \delta_1+\delta_2$ and the case that $d \ge \delta_1+\delta_2$ because of the two thresholds in the DTR policy, see the dashed rectangles in Fig. \ref{fig:markov4} for instance. When $\delta_2+1 \le d \le \delta_1+\delta_2$, we note that the state $\left({1,d}\right)$ can be reached from the states $\left({k^\prime,d^\prime}\right)$, $\forall k^\prime, d^\prime$ with $k^\prime + d^\prime = d$, and $d^\prime \le \delta_2-1$. Therefore, we can obtain that
\begin{equation}\label{area3case1eq1}
\begin{split}
\pi_{1,d} = p\left[ {{\pi _{d,0}} + \sum\limits_{n = 2}^{\delta_2-1} {{\pi _{d - n,n}}} } \right]\mathop { = }\limits^{\left( a \right)}p\left[ {{\pi _{d,0}} + \left( {1 - p} \right)\sum\limits_{n = 2}^{{\delta _2} - 1} {{\pi _{d - 1 - n,n}}} } \right],\delta_2+1 \le d \le \delta_1+\delta_2,\\
\end{split}
\end{equation}
where the equality $a$ is due to the property given in Lemma \ref{property1}. To further simplify (\ref{area3case1eq1}), we use the property given in (\ref{fact3}) for Lemma \ref{property3} to replace the term ${\pi _{d,0}}$. However, we realize that it is difficult to further simplify (\ref{area3case1eq1}) because of the two cases described in (\ref{fact3}). We are thus motivated to omit the term $x\left(1-q\right)^{k-1}$ in (\ref{fact3}) such that ${\pi _{d,0}} \approx \left(1-p\right){\pi _{d-1,0}}$ holds for $d \ge \delta_2$. Note that the adopted approximation is very tight for relatively large value of $q$, and it becomes exact when $q \to 1$ ($0 < x <1$). More importantly, the adopted approximation becomes exact when $\delta_1 \le \delta_2-1$, because $d \ge \delta_2 \ge \delta_1+1$ and, ${\pi _{d,0}} = \left(1-p\right){\pi _{d-1,0}}$ always holds according to (\ref{fact3}). With the above analysis, substitute $\pi_{d,0} \approx {\pi _{d-1,0}} \left(1-p\right)$ into (\ref{area3case1eq1}), we can further derive that
\begin{equation}\label{area3case1eq2}
\pi_{1,d} \approx \left( {1 - p} \right){\pi _{1,d - 1}}, \delta_2+1 \le d \le \delta_1+\delta_2.
\end{equation}

We now turn to the evolution of $\pi_{1,d}$ when $d \ge \delta_2+\delta_1$. Differently, the state $\left({1,d}\right)$ can now be readched from the states $\left({k^\prime,d^\prime}\right)$, $\forall k^\prime, d^\prime$ with $k^\prime + d^\prime = d$, and $d^\prime \le d-\delta_1-1$, see the dashed rectangular depicted in Fig. \ref{fig:markov4} for instance. Similar to (\ref{area3case1eq1}), we have
\begin{equation}\label{area3case1eq3}
\begin{split}
\pi_{1,d}&=p\left[ {{\pi _{d,0}} + \sum\limits_{n = 2}^{d - {\delta _1} - 1} {{\pi _{d - n,n}}} } \right],d \ge \delta_2+\delta_1,\\
&=p\left[ {\left( {1 - p} \right){\pi _{d - 1,0}} + \sum\limits_{n = 2}^{d - {\delta _1} - 2} {{\pi _{d - n,n}}}  + {\pi _{{\delta _1} + 1,d - {\delta _1} - 1}}} \right],d \ge \delta_2+\delta_1,\\
&=p\left[ {\left( {1 - p} \right)\left( {{\pi _{d - 1,0}} + \sum\limits_{n = 2}^{d - {\delta _1} - 2} {{\pi _{d - n - 1,n}}} } \right) + {\pi _{{\delta _1} + 1,d - {\delta _1} - 1}}} \right],d \ge \delta_2+\delta_1,\\
&\mathop { = }\limits^{\left( a \right)}{\left( {1 - p} \right){\pi _{1,d-1}} + p\left( {1 - q} \right)^{\delta_1}{\pi _{1,d - {\delta _1-1}}}},d \ge \delta_2+\delta_1,
\end{split}
\end{equation}
where the equality $a$ for the term ${\pi _{{\delta _1} + 1,d - {\delta _1} - 1}}$ is due to the property given in Lemma \ref{property2}.
\subsection{Proof of Lemma \ref{lemmax}}
We solve $x$ by using the fact that the sum of the stationary distribution for all the states in the MC are equal to $1$. Specifically, we have the equation
\begin{equation}\label{xequation}
\begin{split}
\sum\limits_{\forall k,d} {{\pi _{k,d}}}  & =\sum\limits_{k = 2}^\infty  {{\pi _{k,0}}} + \sum\limits_{k = 1}^{\infty} \sum\limits_{d = 2}^{{\delta _2} - 1}  {{\pi _{k,d}}}  + \sum\limits_{k = 1}^\infty  \sum\limits_{d = {\delta _2}}^\infty  {{\pi _{k,d}}} = 1,\\
\end{split}
\end{equation}
where the three terms $\sum\limits_{k = 2}^\infty  {{\pi _{k,0}}}$, $\sum\limits_{k = 1}^{\infty} \sum\limits_{d = 2}^{{\delta _2} - 1}  {{\pi _{k,d}}}$ and $\sum\limits_{k = 1}^\infty  \sum\limits_{d = {\delta _2}}^\infty  {{\pi _{k,d}}}$ represent the sum of stationary distribution for subspaces $A$, $B$ and $C$, respectively. By using the analytical results given in (8a) and (9a) for the stationary distribution of subspace $A$, the term $\sum\limits_{k = 2}^\infty  {{\pi _{k,0}}}$ can be evaluated as
\begin{equation}\label{app3term1}
\begin{split}
\sum\limits_{k = 2}^\infty  {{\pi _{k,0}}} &= \sum\limits_{k = 2}^{{\delta _1} + 1} {{{{{\left( {1 - p} \right)}^{k - 1}} - {{\left( {1 - q} \right)}^{k - 1}}} \over {q - p}}x}  + \sum\limits_{k = {\delta _1} + 2}^\infty  {{{\left[ {{{\left( {1 - p} \right)}^{{\delta _1}}} - {{\left( {1 - q} \right)}^{{\delta _1}}}} \right]{{\left( {1 - p} \right)}^{k - {\delta _1} - 1}}} \over {q - p}}x} \\
&= {{1 - {{\left( {1 - q} \right)}^{{\delta _1}}}} \over {pq}}x.
\end{split}
\end{equation}
We now turn to the second term in (\ref{xequation}), i.e., $\sum\limits_{k = 1}^{\infty} \sum\limits_{d = 2}^{{\delta _2} - 1}  {{\pi _{k,d}}}$ for subspace $B$ based on (8b) and (9b). When $\delta_1 \ge \delta_2-1$, we can obtain that
\begin{equation}\label{app3term21}
\begin{split}
\sum\limits_{k = 1}^{\infty} \sum\limits_{d = 2}^{{\delta _2} - 1}  {{\pi _{k,d}}} &\mathop { = }\limits^{\left( a \right)} \sum\limits_{d = 2}^{{\delta _2} - 1} {\pi _{1,d}}\sum\limits_{k = 1}^\infty  {{{\left( {1 - p} \right)}^{k - 1}}} \\
&\mathop { = }\limits^{\left( b \right)} \sum\limits_{d = 2}^{{\delta _2} - 1} {{\pi _{1,d}} \over p}\\
&\mathop { = }\limits^{\left( c \right)}\sum\limits_{k = 2}^{{\delta _2} - 1} {{{1 - {{\left( {1 - q} \right)}^{k - 1}}} \over q}} x  \\
& = {{q{\delta _2} - q - \left[ {1 - {{\left( {1 - q} \right)}^{{\delta _2} - 1}}} \right]} \over {{q^2}}}x, \\
\end{split}
\end{equation}
where the equality $a$ is according to the property described in Lemma \ref{property1}. The equality $b$ can be obtained by the sum of geometric sequence, and the equality $c$ is derived by substituting $k=1$ into (8b). When $\delta_1 \le \delta_2-1$, similar to the above derivation and using (9b), the term $\sum\limits_{k = 1}^{\infty} \sum\limits_{d = 2}^{{\delta _2} - 1}  {{\pi _{k,d}}}$ for area $B$ can be calculated as
\begin{equation}\label{app3term22}
\begin{split}
\sum\limits_{k = 1}^{\infty} \sum\limits_{d = 2}^{{\delta _2} - 1}  {{\pi _{k,d}}} ={{q{\delta _2} - q - 1 + \left( {1 - q{\delta _2} + q{\delta _1} + q} \right){{\left( {1 - q} \right)}^{{\delta _1}}}} \over {{q^2}}}x.
\end{split}
\end{equation}
The last term in (\ref{xequation}) $\sum\limits_{k = 1}^\infty  \sum\limits_{d = \delta_2}^\infty  {{\pi _{k,d}}}$ for subspace $C$ can be calculated as
\begin{equation}\label{app3term3}
\begin{split}
\sum\limits_{k = 1}^\infty  \sum\limits_{d = \delta_2}^\infty  {{\pi _{k,d}}} &\mathop { = }\limits^{\left( a \right)} \sum\limits_{d = {\delta _2}}^\infty{\pi _{1,d}} \left[\sum\limits_{k = 1}^{{\delta _1} + 1} {{{\left( {1 - q} \right)}^{{k-1}}}}  + {\left( {1 - q} \right)^{{\delta _1}}}\sum\limits_{k = \delta_1+2}^\infty  {{{\left( {1 - p} \right)}^{k-\delta_1-1}}}\right]\\
& \mathop { = }\limits^{\left( b \right)} {{p + \left( {q - p} \right){{\left( {1 - q} \right)}^{{\delta _1}}}} \over {pq}} \sum\limits_{d = {\delta _2}}^\infty{\pi _{1,d}}\\
& \mathop { = }\limits^{\left( c \right)}{{p + \left( {q - p} \right){{\left( {1 - q} \right)}^{{\delta _1}}}} \over {p{q^2}}}x,
\end{split}
\end{equation}
where the equality $a$ can be obtained by Lemma \ref{property2}. The equality $b$ is derived by simplifying the summation terms within the square brackets by the sum of geometric sequence. The equality $c$ is from the analysis given below (\ref{col01equation}) that $\sum\limits_{d = {\delta _2}}^\infty  {{\pi _{1,d}}}   = {x \over q}$. Substitute (\ref{app3term1}), (\ref{app3term21}), (\ref{app3term22}) and (\ref{app3term3}) into (\ref{xequation}), we can solve $x$ and obtain the desired results given in (\ref{x}) for the two cases.
\section{Proof of Lemmas in Section V}\label{AppsectionV}
\subsection{Proof of Lemma \ref{AoItermlemma}}\label{app5}
Based on the stationary distribution of subspace $A$ given in (8a) and (9a), the term $\sum\limits_{k = 2}^\infty  {k{\pi _{k,0}}}$ can be calculated as
\begin{equation}
\begin{split}
\sum\limits_{k = 2}^\infty  {k{\pi _{k,0}}}&=\sum\limits_{k = 2}^{{\delta _1} + 1} {k{{{{\left( {1 - p} \right)}^{k - 1}} - {{\left( {1 - q} \right)}^{k - 1}}} \over {q - p}}x}  + \sum\limits_{k = {\delta _1} + 2}^\infty  {k{{\left[ {{{\left( {1 - p} \right)}^{{\delta _1}}} - {{\left( {1 - q} \right)}^{{\delta _1}}}} \right]{{\left( {1 - p} \right)}^{k - {\delta _1} - 1}}} \over {q - p}}x}\\
&\mathop { = }\limits^{\left( a \right)} {{\left( {p + q} \right)\left[ {1 - {{\left( {1 - q} \right)}^{{\delta _1}}}} \right] - pq{\delta _1}{{\left( {1 - q} \right)}^{{\delta _1}}}} \over {{p^2}{q^2}}}x,
\end{split}
\end{equation}
where the equality $a$ can be obtained by the sum of geometric sequence.
\subsection{Proof of Lemma \ref{AoIterm2lemma}}\label{app6}
In order to evaluate the term $\sum\limits_{k = 1}^\infty  {\sum\limits_{d = 2}^{{\delta _2} - 1} {\left( {k + d} \right){\pi _{k,d}}} }$, we first consider the term $\sum\limits_{k = 1}^\infty  {\left( {k + d} \right){\pi _{k,d}}}$ with $2 \le d \le \delta_2-1$, and it is given by
\begin{equation}\label{Aiterm2}
\begin{split}
\sum\limits_{k = 1}^\infty  {\left( {k + d} \right){\pi _{k,d}}}, 2 \le d \le \delta_2-1&\mathop { = }\limits^{\left( a \right)} \sum\limits_{k = 1}^\infty  {\left( {k + d} \right)\pi_{1,d}\left(1-p\right)^{k-1}},2 \le d \le \delta_2-1, \\
&\mathop { = }\limits^{\left( b \right)}{{pd + 1} \over {{p^2}}}\pi_{1,d},2 \le d \le \delta_2-1,
\end{split}
\end{equation}
where the equality $a$ is according to the property given in Lemma \ref{property1}, and the equality $b$ can be obtained by the sum of series $\sum\limits_{n = 0}^\infty  {\left( {n + a} \right) {{\left( {1 - p} \right)}^n}}  = {{pa + 1} \over {{p^2}}}$ provided by WolframAlpha \cite{wolfram}. With the result given in (\ref{Aiterm2}), the term $\sum\limits_{k = 1}^\infty  {\sum\limits_{d = 2}^{{\delta _2} - 1} {\left( {k + d} \right){\pi _{k,d}}} }$ can be further calculated as
\begin{equation}\label{AoIterm2eq}
\begin{split}
\sum\limits_{k = 1}^\infty  {\sum\limits_{d = 2}^{{\delta _2} - 1} {\left( {k + d} \right){\pi _{k,d}}} } &= {1 \over {{p^2}}}\mathop \sum \limits_{d = 2}^{{\delta _2} - 1} {\pi _{1,d}} + {1 \over p}\mathop \sum \limits_{d = 2}^{{\delta _2} - 1} d{\pi _{1,d}}.\\
\end{split}
\end{equation}
The term ${1 \over {{p^2}}}\mathop \sum \limits_{d = 2}^{{\delta _2} - 1} {\pi _{1,d}}$ in (\ref{AoIterm2eq}) is readily given for the two cases $\delta_1 \ge \delta_2-1$ and $\delta_1 \le \delta_2-1$ by using (\ref{app3term21}) and (\ref{app3term22}), respectively. We now focus on the term ${1 \over p}\mathop \sum \limits_{d = 2}^{{\delta _2} - 1} d{\pi _{1,d}}$ in (\ref{AoIterm2eq}). When $\delta_1 \ge \delta_2-1$, the term ${1 \over p}\mathop \sum \limits_{d = 2}^{{\delta _2} - 1} d{\pi _{1,d}}$ can be evaluated as
\begin{equation}\label{AoIterm2case1eq}
\begin{split}
 {1 \over p}\mathop \sum \limits_{d = 2}^{{\delta _2} - 1} d{\pi _{1,d}} &\mathop { = }\limits^{\left( a \right)}{x \over q}\sum\limits_{d = 2}^{{\delta _2} - 1} d \left[ {1 - {{\left( {1 - q} \right)}^{d - 1}}} \right]\\
& \mathop { = }\limits^{\left( b \right)} {{\left( {{\delta _2} + 1} \right)\left( {{\delta _2} - 2} \right)} \over {2q}}x - {{1 - {q^2} - \left( {1 - q + q{\delta _2}} \right){{\left( {1 - q} \right)}^{{\delta _2} - 1}}} \over {{q^3}}}x,
\end{split}
\end{equation}
where the equality $a$ can be obtained by (8b), and the equality $b$ is derived from the sum of arithmetic sequence and the sum of geometric sequence. When $\delta_1 \le \delta_2-1$, similar to (\ref{AoIterm2case1eq}) and using (9b), we can calculate that
\begin{equation}\label{AoIterm2case2eq}
\begin{split}
{1 \over p}\sum\limits_{d = 2}^{{\delta _2} - 1} d {\pi _{1,d}}  &={{\left( {{\delta _2} + 1} \right)\left( {{\delta _2} - 2} \right)} \over {2q}}x - {{1 - {q^2} - \left( {1 + q + q{\delta _1}} \right){{\left( {1 - q} \right)}^{{\delta _1} + 1}}} \over {{q^3}}}x\\
& \quad  - {{{{\left( {1 - q} \right)}^{{\delta _1}}}\left( {{\delta _1} + {\delta _2} + 1} \right)\left( {{\delta _2} - {\delta _1} - 2} \right)} \over {2q}}x.
\end{split}
\end{equation}
Substitute (\ref{app3term21}), (\ref{app3term22}), (\ref{AoIterm2case1eq}) and (\ref{AoIterm2case2eq}) into (\ref{AoIterm2eq}), after some manipulation, we can obtain the desired results given in Lemma \ref{AoIterm2lemma} and here completes the proof.
\subsection{Proof of Lemma \ref{AoIterm3lemma}}\label{app8}
In order to calculate the term $\sum\limits_{k = 1}^\infty  {\sum\limits_{d = {\delta _2}}^\infty  {\left( {k + d} \right){\pi _{k,d}}} }$, we first derive the term $\sum\limits_{k = 1}^\infty {\left( {k + d} \right){\pi _{k,d}}}$, and it is given by
\begin{equation}\label{term3lemmaeq1}
\begin{split}
\sum\limits_{k = 1}^\infty {\left( {k + d} \right){\pi _{k,d}}}&  \mathop { = }\limits^{\left( a \right)} \sum\limits_{k = 1}^{{\delta _1} + 1} {\left( {k + d} \right)} {\left( {1 - q} \right)^{k - 1}}{\pi _{1,d}} + {\left( {1 - q} \right)^{{\delta _1}}}\sum\limits_{k = {\delta _1} + 2}^\infty  {\left( {k + d} \right)} {\left( {1 - p} \right)^{k - {\delta _1} - 1}}{\pi _{1,d}}\\
&= {\pi _{1,d}}\underbrace {\left[ {\sum\limits_{k = 1}^{{\delta _1} + 1} k {{\left( {1 - q} \right)}^{k - 1}} + {{\left( {1 - q} \right)}^{{\delta _1}}}\sum\limits_{k = {\delta _1} + 2}^\infty  k {{\left( {1 - p} \right)}^{k - {\delta _1} - 1}}} \right]}_{{\Theta _1}}\\
&\quad+d{\pi _{1,d}}\underbrace {\left[ {\sum\limits_{k = 1}^{{\delta _1} + 1} {{{\left( {1 - q} \right)}^{k - 1}}}  + {{\left( {1 - q} \right)}^{{\delta _1}}}\sum\limits_{k = {\delta _1} + 2}^\infty  {{{\left( {1 - p} \right)}^{k - {\delta _1} - 1}}} } \right]}_{{\Theta _2}},
\end{split}
\end{equation}
where the equality $a$ is according to the fact given in Lemma \ref{property2}. By using the sum of geometric sequence, the terms ${{\Theta _1}}$ and ${{\Theta _2}}$ in (\ref{term3lemmaeq1}) are given by
\begin{equation}\label{theta1}
{{\Theta _1}}={\left( {1 - q} \right)^{{\delta _1}}}\left( {{{{\delta _1}} \over p} + {1 \over {{p^2}}} - {{{\delta _1} } \over q} - {1 \over {{q^2}}}} \right)  + {1 \over {{q^2}}},
\end{equation}
\begin{equation}\label{theta2}
{{\Theta _2}} = \left( {{1 \over p} - {1 \over q}} \right){{\left( {1 - q} \right)}^{{\delta _1}}} + {1 \over q}.
\end{equation}
With the results given in (\ref{term3lemmaeq1}), (\ref{theta1}) and (\ref{theta2}), the term $\sum\limits_{k = 1}^\infty  {\sum\limits_{d = {\delta _2}}^\infty  {\left( {k + d} \right){\pi _{k,d}}} }$ can now be further evaluated as
\begin{equation}\label{term3lemmaeq2}
\begin{split}
\sum\limits_{k = 1}^\infty  {\sum\limits_{d = {\delta _2}}^\infty  {\left( {k + d} \right){\pi _{k,d}}} } &= {\Theta _1}\sum\limits_{d = {\delta _2}}^\infty  {{\pi _{1,d}}}  + {\Theta _2}\sum\limits_{d = {\delta _2}}^\infty  {d{\pi _{1,d}}}\mathop { = }\limits^{\left( a \right)}{{{\Theta _1} + {\Theta _2}{\delta _2}} \over q}x+{\Theta _2}\sum\limits_{d = {\delta _2}}^\infty  {\left(d-{\delta _2}\right){\pi _{1,d}}},\\
\end{split}
\end{equation}
where the equality $a$ can be obtained by the analysis given below (\ref{col01equation}) that $\sum\limits_{d = {\delta _2}}^\infty  {{\pi _{1,d}} = } {x \over q}$. ${{\Theta _1}}$ and ${{\Theta _2}}$ in (\ref{term3lemmaeq2}) has been derived in (\ref{theta1}) and (\ref{theta2}), we now consider the only unknown term $\sum\limits_{d = {\delta _2}}^\infty  {\left(d-{\delta _2}\right){\pi _{1,d}}}$ in (\ref{term3lemmaeq2}). By using the stationary distribution calculated in (8c) and (9c) for subspace $C$, the term $\sum\limits_{d = {\delta _2}}^\infty  {\left(d-{\delta _2}\right){\pi _{1,d}}}$  can be evaluated as
\begin{equation}\label{term3lemmaeq3}
\begin{split}
\sum\limits_{d = {\delta _2}}^\infty  {\left(d-{\delta _2}\right){\pi _{1,d}}} &\mathop { = }\limits^{\left( a \right)}{\pi _{1,{\delta _2}}}\sum\limits_{l = 0}^\infty  {\sum\limits_{N = 1}^\infty  {\left[ {l\left( {{\delta _1} + 1} \right) + N - 1} \right]} } {{\left( {N + l - 1} \right)!} \over {l!\left( {N - 1} \right)!}}{\left[ {p{{\left( {1 - q} \right)}^{{\delta _1}}}} \right]^l}{\left( {1 - p} \right)^{N - 1}}\\
&\mathop { = }\limits^{\left( b \right)}{\pi _{1,{\delta _2}}}\sum\limits_{l = 0}^\infty  {{{\left( {1 - q} \right)}^{l{\delta _1}}}{{\left( {p{\delta _1}l + l - p + 1} \right)} \over {{p^2}}}} \\
& \mathop { = }\limits^{\left( c \right)} {\pi _{1,{\delta _2}}}\left[{{{\delta _1}} \over p}\sum\limits_{l = 0}^\infty  l {\left( {1 - q} \right)^{l{\delta _1}}} + {1 \over {{p^2}}}\sum\limits_{l = 0}^\infty  {\left( {l + 1} \right)} {\left( {1 - q} \right)^{l{\delta _1}}} - {1 \over p}\sum\limits_{l = 0}^\infty  {{{\left( {1 - q} \right)}^{l{\delta _1}}}}\right]\\
&\mathop { = }\limits^{\left( d \right)}{\pi _{1,{\delta _2}}}\left\{{{{\delta _1}{{\left( {1 - q} \right)}^{{\delta _1}}}} \over {p{{\left[ {1 - {{\left( {1 - q} \right)}^{{\delta _1}}}} \right]}^2}}} + {1 \over {{p^2}{{\left[ {1 - {{\left( {1 - q} \right)}^{{\delta _1}}}} \right]}^2}}} - {1 \over {p\left[ {1 - {{\left( {1 - q} \right)}^{{\delta _1}}}} \right]}}\right\},
\end{split}
\end{equation}
where the equality $a$ is obtained by substituting $k=1$ into (8c) and (9c). Besides, from (\ref{n}) and (\ref{m}), we have $d-\delta_2 =  n\left(\delta_1+1\right)+m = {l\left( {{\delta _1} + 1} \right) + N - 1}$. The equality $b$ in (\ref{term3lemmaeq3}) is because of the following sum of series calculated on WolframAlpha \cite{wolfram}
\begin{equation}
\sum\limits_{N = 1}^\infty  {\left[ {l\left( {{\delta _1} + 1} \right) + N - 1} \right]}  {{\left( {N + l - 1} \right)!} \over {l!\left( {N - 1} \right)!}}{p^l}{\left( {1 - p} \right)^{N - 1}} = {{\left( {p{\delta _1}l + l - p + 1} \right)} \over {{p^2}}}.
\end{equation}
The equality $c$ in (\ref{term3lemmaeq3}) is derived by some algebraic manipulation. The equality $d$ in (\ref{term3lemmaeq3}) is derived based on the sum of geometric sequence and the sum of arithmetic sequence. Substitute (\ref{theta1}), (\ref{theta2}), (\ref{term3lemmaeq3}) into (\ref{term3lemmaeq2}), we obtain the desired results given in Lemma \ref{AoIterm3lemma}.

%\section*{Acknowledgment}
\ifCLASSOPTIONcaptionsoff
  \newpage
\fi

\bibliographystyle{IEEEtran}
\bibliography{References}

\begin{thebibliography}{10}
\providecommand{\url}[1]{#1}
\csname url@samestyle\endcsname
\providecommand{\newblock}{\relax}
\providecommand{\bibinfo}[2]{#2}
\providecommand{\BIBentrySTDinterwordspacing}{\spaceskip=0pt\relax}
\providecommand{\BIBentryALTinterwordstretchfactor}{4}
\providecommand{\BIBentryALTinterwordspacing}{\spaceskip=\fontdimen2\font plus
\BIBentryALTinterwordstretchfactor\fontdimen3\font minus
  \fontdimen4\font\relax}
\providecommand{\BIBforeignlanguage}[2]{{%
\expandafter\ifx\csname l@#1\endcsname\relax
\typeout{** WARNING: IEEEtran.bst: No hyphenation pattern has been}%
\typeout{** loaded for the language `#1'. Using the pattern for}%
\typeout{** the default language instead.}%
\else
\language=\csname l@#1\endcsname
\fi
#2}}
\providecommand{\BIBdecl}{\relax}
\BIBdecl

\bibitem{IoTbackground}
F.~Lin, C.~Chen, N.~Zhang, X.~Guan, and X.~Shen, ``Autonomous channel
  switching: Towards efficient spectrum sharing for industrial wireless sensor
  networks,'' \emph{IEEE Internet of Things Journal}, vol.~3, no.~2, pp.
  231--243, April 2016.

\bibitem{AoIbackground}
S.~{Kaul}, R.~{Yates}, and M.~{Gruteser}, ``Real-time status: How often should
  one update?'' in \emph{Proc. 2012 IEEE INFOCOM}, 2012, pp. 2731--2735.

\bibitem{AoI_LCFS}
S.~K. Kaul, R.~D. Yates, and M.~Gruteser, ``Status updates through queues,'' in
  \emph{Proc. 46th Annual Conference on Information Sciences and Systems
  ({C}{I}{S}{S})}, 2012, pp. 1--6.

\bibitem{AEphremides_management}
M.~{Costa}, M.~{Codreanu}, and A.~{Ephremides}, ``On the age of information in
  status update systems with packet management,'' \emph{IEEE Transactions on
  Information Theory}, vol.~62, no.~4, pp. 1897--1910, Feb. 2016.

\bibitem{AEphremides_MM12_deadline}
C.~{Kam}, S.~{Kompella}, G.~D. {Nguyen}, J.~E. {Wieselthier}, and
  A.~{Ephremides}, ``On the age of information with packet deadlines,''
  \emph{IEEE Transactions on Information Theory}, vol.~64, no.~9, pp.
  6419--6428, Sep. 2018.

\bibitem{Updateorwait}
Y.~{Sun}, E.~{Uysal-Biyikoglu}, R.~D. {Yates}, C.~E. {Koksal}, and N.~B.
  {Shroff}, ``Update or wait: How to keep your data fresh,'' \emph{IEEE
  Transactions on Information Theory}, vol.~63, no.~11, pp. 7492--7508, August
  2017.

\bibitem{AValehi_CR}
A.~{Valehi} and A.~{Razi}, ``Maximizing energy efficiency of cognitive wireless
  sensor networks with constrained age of information,'' \emph{IEEE
  Transactions on Cognitive Communications and Networking}, vol.~3, no.~4, pp.
  643--654, Dec. 2017.

\bibitem{SLeng_CR}
S.~{Leng} and A.~{Yener}, ``Age of information minimization for an energy
  harvesting cognitive radio,'' \emph{IEEE Transactions on Cognitive
  Communications and Networking}, vol.~5, no.~2, pp. 427--439, June 2019.

\bibitem{YifanCR}
Y.~{Gu}, H.~{Chen}, C.~{Zhai}, Y.~{Li}, and B.~{Vucetic}, ``Minimizing age of
  information in cognitive radio-based iot systems: Underlay or overlay?''
  \emph{IEEE Internet of Things Journal}, vol.~6, no.~6, pp. 10\,273--10\,288,
  August 2019.

\bibitem{QianCR}
Q.~{Wang}, H.~{Chen}, Y.~{Gu}, Y.~{Li}, and B.~{Vucetic}, ``Minimizing the age
  of information of cognitive radio-based iot systems under a collision
  constraint,'' \emph{IEEE Transactions on Wireless Communications}, vol.~19,
  no.~12, pp. 8054--8067, Dec. 2020.

\bibitem{Ryates_multiplesources}
R.~{Yates} and S.~{Kaul}, ``The age of information: Real-time status updating
  by multiple sources,'' \emph{IEEE Transactions on Information Theory},
  vol.~65, no.~3, pp. 1807--1827, March 2019.

\bibitem{R.Talak-Centralized-UnknownCSI-MAC}
R.~{Talak}, I.~{Kadota}, S.~{Karaman}, and E.~{Modiano}, ``Scheduling policies
  for age minimization in wireless networks with unknown channel state,'' in
  \emph{Proc. IEEE International Symposium on Information Theory (ISIT)}, 2018,
  pp. 2564--2568.

\bibitem{I.Kadota-Centralized-MAC}
I.~{Kadota}, A.~{Sinha}, E.~{Uysal-Biyikoglu}, R.~{Singh}, and E.~{Modiano},
  ``Scheduling policies for minimizing age of information in broadcast wireless
  networks,'' \emph{IEEE/ACM Transactions on Networking}, vol.~26, no.~6, pp.
  2637--2650, December 2018.

\bibitem{Z.Jiang-Decentralized}
Z.~{Jiang}, B.~{Krishnamachari}, X.~{Zheng}, S.~{Zhou}, and Z.~{Niu}, ``Timely
  status update in wireless uplinks: Analytical solutions with asymptotic
  optimality,'' \emph{IEEE Internet of Things Journal}, vol.~6, no.~2, pp.
  3885--3898, April 2019.

\bibitem{S.Kaul-Distributed-Centralized-MAC}
S.~K. {Kaul} and R.~D. {Yates}, ``Status updates over unreliable multiaccess
  channels,'' in \emph{Proc. 2017 IEEE International Symposium on Information
  Theory (ISIT)}, 2017, pp. 331--335.

\bibitem{yifanmac}
H.~{Chen}, Y.~{Gu}, and S.~C. {Liew}, ``Age-of-information dependent random
  access for massive iot networks,'' in \emph{Proc. IEEE INFOCOM 2020 - IEEE
  Conference on Computer Communications Workshops (INFOCOM WKSHPS)}, 2020, pp.
  930--935.

\bibitem{Rnew}
Y.~{Sun} and B.~{Cyr}, ``Sampling for data freshness optimization: Non-linear
  age functions,'' \emph{Journal of Communications and Networks}, vol.~21,
  no.~3, pp. 204--219, July 2019.

\bibitem{AoI_EH1}
R.~D. {Yates}, ``Lazy is timely: Status updates by an energy harvesting
  source,'' in \emph{Proc. 2015 IEEE International Symposium on Information
  Theory (ISIT)}, 2015, pp. 3008--3012.

\bibitem{AoI_EH2}
B.~T. {Bacinoglu}, E.~T. {Ceran}, and E.~{Uysal-Biyikoglu}, ``Age of
  information under energy replenishment constraints,'' in \emph{Proc. 2015
  Information Theory and Applications Workshop (ITA)}, 2015, pp. 25--31.

\bibitem{AoI_EH3}
A.~{Arafa} and S.~{Ulukus}, ``Age minimization in energy harvesting
  communications: Energy-controlled delays,'' in \emph{Proc. 2017 51st Asilomar
  Conference on Signals, Systems, and Computers}, 2017, pp. 1801--1805.

\bibitem{AoI_EH4}
S.~{Feng} and J.~{Yang}, ``Optimal status updating for an energy harvesting
  sensor with a noisy channel,'' in \emph{Proc. IEEE INFOCOM 2018 - IEEE
  Conference on Computer Communications Workshops (INFOCOM WKSHPS)}, 2018, pp.
  348--353.

\bibitem{AoI_EH5}
X.~{Wu}, J.~{Yang}, and J.~{Wu}, ``Optimal status update for age of information
  minimization with an energy harvesting source,'' \emph{IEEE Transactions on
  Green Communications and Networking}, vol.~2, no.~1, pp. 193--204, March
  2018.

\bibitem{AoI_EH6}
I.~{Krikidis}, ``Average age of information in wireless powered sensor
  networks,'' \emph{IEEE Wireless Communications Letters}, vol.~8, no.~2, pp.
  628--631, April 2019.

\bibitem{YIFANAOI}
Y.~{Gu}, H.~{Chen}, Y.~{Zhou}, Y.~{Li}, and B.~{Vucetic}, ``Timely status
  update in internet of things monitoring systems: An age-energy tradeoff,''
  \emph{IEEE Internet of Things Journal}, vol.~6, no.~3, pp. 5324--5335, June
  2019.

\bibitem{AoI_HARQ}
E.~T. {Ceran}, D.~{G{\"u}n{\"u}z}, and A.~{Gy{\"o}rgy}, ``Average age of
  information with hybrid arq under a resource constraint,'' \emph{IEEE
  Transactions on Wireless Communications}, vol.~18, no.~3, pp. 1900--1913,
  March 2019.

\bibitem{BzhouCMDP}
B.~{Zhou} and W.~{Saad}, ``Joint status sampling and updating for minimizing
  age of information in the internet of things,'' \emph{IEEE Transactions on
  Communications}, vol.~67, no.~11, pp. 7468--7482, Nov. 2019.

\bibitem{multihop1}
A.~M. {Bedewy}, Y.~{Sun}, and N.~B. {Shroff}, ``Age-optimal information updates
  in multihop networks,'' in \emph{Proc. 2017 IEEE International Symposium on
  Information Theory (ISIT)}, 2017, pp. 576--580.

\bibitem{multihop2}
R.~{Talak}, S.~{Karaman}, and E.~{Modiano}, ``Minimizing age-of-information in
  multi-hop wireless networks,'' in \emph{Proc. 2017 55th Annual Allerton
  Conference on Communication, Control, and Computing (Allerton)}, 2017, pp.
  486--493.

\bibitem{multihop3}
C.~{Joo} and A.~{Eryilmaz}, ``Wireless scheduling for information freshness and
  synchrony: Drift-based design and heavy-traffic analysis,'' \emph{IEEE/ACM
  Transactions on Networking}, vol.~26, no.~6, pp. 2556--2568, Sept. 2018.

\bibitem{multihop4}
S.~{Farazi}, A.~G. {Klein}, and D.~R. {Brown}, ``Fundamental bounds on the age
  of information in multi-hop global status update networks,'' \emph{Journal of
  Communications and Networks}, vol.~21, no.~3, pp. 268--279, July 2019.

\bibitem{multihop5}
J.~Selen, Y.~Nazarathy, L.~L.~H. Andrew, and H.~L. Vu, ``The age of information
  in gossip networks,'' in \emph{Proc. Analytical and Stochastic Modeling
  Techniques and Applications}, 2013, pp. 364--379.

\bibitem{simplerelay}
A.~{Maatouk}, M.~{Assaad}, and A.~{Ephremides}, ``The age of updates in a
  simple relay network,'' in \emph{Proc. 2018 IEEE Information Theory Workshop
  (ITW)}, 2018, pp. 1--5.

\bibitem{bobrelay}
B.~Li, H.~Chen, Y.~Zhou, and Y.~Li, ``Age-oriented opportunistic relaying in
  cooperative status update systems with stochastic arrivals,'' \emph{arXiv
  e-prints}, p. arXiv:2001.04084, Jan. 2020.

\bibitem{R1}
J.~P. {Champati}, H.~{Al-Zubaidy}, and J.~{Gross}, ``Statistical guarantee
  optimization for {A}o{I} in single-hop and two-hop {F}{C}{F}{S} systems with
  periodic arrivals,'' \emph{IEEE Transactions on Communications}, pp. 1--1,
  Sept. 2020.

\bibitem{R2}
M.~{Moradian} and A.~{Dadlani}, ``Age of information in scheduled wireless
  relay networks,'' in \emph{in Proc. 2020 IEEE Wireless Communications and
  Networking Conference (WCNC)}, 2020, pp. 1--6.

\bibitem{R3}
C.~{Kam}, J.~P. {Molnar}, and S.~{Kompella}, ``Age of information for queues in
  tandem,'' in \emph{in Proc. MILCOM 2018 - 2018 IEEE Military Communications
  Conference (MILCOM)}, 2018, pp. 1--6.

\bibitem{R4}
F.~{Chiariotti}, O.~{Vikhrova}, B.~{Soret}, and P.~{Popovski}, ``{Peak Age of
  Information Distribution for Edge Computing with Wireless Links},''
  \emph{arXiv e-prints}, p. arXiv:2004.05088, Apr. 2020.

\bibitem{EHtwohop}
A.~{Arafa} and S.~{Ulukus}, ``Timely updates in energy harvesting two-hop
  networks: Offline and online policies,'' \emph{IEEE Transactions on Wireless
  Communications}, vol.~18, no.~8, pp. 4017--4030, June 2019.

\bibitem{sennott1993constrained}
L.~I. Sennott, ``Constrained average cost markov decision chains,''
  \emph{Probability in the Engineering and Informational Sciences}, vol.~7,
  no.~1, pp. 69--83, 1993.

\bibitem{spall2005introduction}
J.~C. Spall, \emph{Introduction to stochastic search and optimization:
  estimation, simulation, and control}.\hskip 1em plus 0.5em minus 0.4em\relax
  John Wiley \& Sons, 2005, vol.~65.

\bibitem{puterman2014markov}
M.~L. Puterman, \emph{Markov Decision Processes.: Discrete Stochastic Dynamic
  Programming}.\hskip 1em plus 0.5em minus 0.4em\relax John Wiley \& Sons,
  2014.

\bibitem{stochasticbook}
\BIBentryALTinterwordspacing
L.~Sennott, \emph{Stochastic Dynamic Programming and the Control of Queueing
  Systems}, ser. Wiley Series in Probability and Statistics.\hskip 1em plus
  0.5em minus 0.4em\relax Wiley, 2009. [Online]. Available:
  \url{https://books.google.com.au/books?id=ene1SIfC1ckC}
\BIBentrySTDinterwordspacing

\bibitem{wolfram}
\BIBentryALTinterwordspacing
``Wolframalpha.com. 2020. - wolfram|alpha.'' accessed: 2020-07-04. [Online].
  Available: \url{https://www.wolframalpha.com/input/}
\BIBentrySTDinterwordspacing

\end{thebibliography}
\end{document}